\documentclass[11pt]{article}
\usepackage
{
        amssymb,
        amsmath,
        enumitem
}

\usepackage{amsthm, fullpage}
\usepackage{times}
\usepackage[margin=2.5cm]{geometry}
\usepackage[usenames,dvipsnames]{xcolor}
\usepackage{nicefrac}

\makeatletter
\renewcommand{\paragraph}{%
   \@startsection{paragraph}{4}%
   {\z@}{2.5ex \@plus 1ex \@minus .2ex}{-1em}%
   {\normalfont\normalsize\bfseries}%
}
\makeatother

\DeclareMathOperator {\budget}  {budget}
\DeclareMathOperator {\totalBudget}  {total-budget}
\DeclareMathOperator {\extraBudget}  {extra-budget}
\DeclareMathOperator {\polylog}  {polylog}

\DeclareMathOperator {\short}  {short}

\DeclareMathOperator {\Var}  {Var}

\newcommand {\Given} {\;|\;}

\newcommand {\roundup}   [1] {{\lceil {#1} \rceil}}
\newcommand {\rounddown} [1] {{\lfloor {#1} \rfloor}}

\newcommand {\set}   [1] {\left\{ #1 \right\}}

\newcommand {\bLog}      {\log_2}
\newcommand {\Exp}       {\mathbb{E}}

\newcommand{\given}{\mid}

\newcommand{\OpenFrame}{\rule{0pt}{12pt} \hrule height 0.8pt \rule{0pt}{1pt} \hrule height 0.4pt \rule{0pt}{6pt}}
\newcommand{\CloseFrame}{\rule{0pt}{1pt}\hrule height 0.4pt \rule{0pt}{1pt} \hrule height 0.8pt \rule{0pt}{12pt}}

\newcommand{\CloseFigureFrame}{\vspace{-10pt}\rule{0pt}{1pt}\hrule height 0.4pt \rule{0pt}{1pt} \hrule height 0.8pt \rule{0pt}{12pt}\vspace{-15pt}}

\newcommand {\bbZ}    {\mathbb{Z}}

\newcommand {\bbR}    {\mathbb{R}}

\newcommand {\calE}   {{\cal{E}}}
\newcommand {\calS}   {{\cal{S}}}

\newcommand {\calU}   {{\cal{U}}}

\newcommand {\calA}   {{\cal{A}}}

\newcommand {\calX}   {{\cal{X}}}

\DeclareMathOperator {\Ball}{Ball}

\DeclareMathOperator {\sdpcost}{sdp-cost}

\newcommand {\VHgeqB}   {V_H^{\geq \beta d}}
\newcommand {\VGleqA}   {V_G^{\leq \alpha d}}

\newtheorem{theorem}{Theorem}[section]
\newtheorem{lemma}[theorem]{Lemma}

\newtheorem{claim}[theorem]{Claim}
\newtheorem{corollary}[theorem]{Corollary}
\newtheorem{definition}[theorem]{Definition}

\newtheorem{property}{Property}

\title{Constant Factor Approximation for Balanced Cut in \\ the PIE Model} 
\author{Konstantin Makarychev\\Microsoft
\and Yury Makarychev\thanks{Supported by NSF CAREER award CCF-1150062 and NSF grant IIS-1302662. Part of this work was done when the author was visiting Microsoft Research.}\\TTIC
\and Aravindan Vijayaraghavan\thanks{Supported by the Simons Postdoctoral Fellowship.}\\CMU}
\date{}
\begin{document}
\maketitle
\begin{abstract}
We propose and study a new semi-random semi-adversarial model for Balanced Cut,
\textit{a planted model with permutation-invariant random edges} (PIE).
Our model is much more general than planted models considered previously.
Consider a set of vertices $V$ partitioned into two clusters $L$ and $R$ of equal size.
Let $G$ be an arbitrary graph on $V$ with no edges between $L$ and $R$.
Let $E_{\text{random}}$ be a set of edges sampled from an arbitrary \textit{permutation-invariant} 
distribution (a distribution that is invariant under permutation of vertices in $L$ and in $R$).
Then we say that $G + E_{\text{random}}$ is a graph with permutation-invariant random edges.

We present an approximation algorithm for the Balanced Cut problem that finds a balanced cut
of cost $O(|E_{\text{random}}|) + n\polylog(n)$ in this model. In the regime when $|E_{\text{random}}| = \Omega(n\polylog(n))$,
this is a constant factor approximation with respect to the cost of the planted cut.
\end{abstract}

\section{Introduction}\label{sec:intro}
Combinatorial optimization problems arise in many areas of science and engineering.
Many of them are $NP$-hard and cannot 
be solved exactly unless $P=NP$. 
What algorithms should we use to solve them? 
There has been a lot of research in theoretical computer science
dedicated to this question. Most research has focused on designing and analyzing approximation algorithms for the worst-case model,
in which
we do not make any assumptions on what the input instances are. 
While this model is very general, algorithms for the worst-case model do not exploit properties that instances we encounter in practice have.
Indeed, as empirical evidence suggests, real-life instances are usually much easier than worst-case instances,
and practitioners often get much better approximation guarantees in real life than it is theoretically possible in the worst-case model.
Thus it is very important to develop a model for real-life instances that will allow us to 
design approximation algorithms that provably work well in practice and outperform known algorithms designed for the worst case.
Several such models have been considered in the literature since the early 80's: e.g.
the random planted cut model~\cite{BCLS84,DF87,B88,JS,DI98,McSherry,CK01,coja}, 
semi-random models~\cite{FK,SVM02,MMV}, 
and stable models~\cite{BBG,Awasthietal,BL,ABBSV,BM,BDLS,MMVstable}.

In this paper, we propose a new very general model ``planted model with permutation-invariant random edges''.
We believe that this model captures well many properties of real-life instances. In particular, 
we argue below that our model is consistent with social network formation models studied in social sciences.
We present an approximation algorithm for the Balanced Cut problem. Balanced Cut
is one of the most basic and well-studied graph partitioning problems.
The problem does not admit a constant factor approximation in the worst-case as was shown by 
Raghavendra, Steurer, and Tulsiani~\cite{RST} (assuming the Small Set Expansion Conjecture).
The best known algorithm 
for Balanced Cut by Arora, Rao, and Vazirani~\cite{ARV} gives $O(\sqrt{\log n})$ approximation. 
In contrast, our algorithm gives a constant factor approximation with respect to the size of the planted cut in our model 
(if some conditions hold, see below).

We start with recalling the classical planted cut model of  Bui, Chaudhuri, Leighton and Sipser~\cite{BCLS84} and Dyer and Frieze~\cite{DF87}. In this model, we generate a random graph $F$ as follows. 
Let $p$ and $q < p$ be two numbers between $0$ and $1$. 
We take two disjoint $G(n/2,p)$ graphs  $G_1=(L,E_1)$ and $G_2=(R,E_2)$. We connect every two vertices $x \in L$ and $y\in R$
with probability $q$; our random choices for all pairs of vertices $(x,y)$ are independent. We obtain a graph $F$.
We call sets $L$ and $R$ \textit{clusters} and say that $(L,R)$ is the \textit{planted cut}. 
We refer to 
the edges added at the second step as \textit{random} edges. 
In this model, we can find  the planted cut $(L,R)$ w.h.p. given the graph $F$ 
(under some assumptions on $p$ and $q$)~\cite{BCLS84,DF87,B88,coja}.

In our model, graphs $G_1$ and $G_2$ can be \textit{arbitrary} graphs. The set of random edges is sampled from an \textit{arbitrary}
\textit{permu\-tation-invariant} distribution (a distribution is permutation-invariant if it is
invariant under permutation of vertices in $L$ and $R$). We do not make any assumptions on the distribution
(aside from it being permutation-invariant). 
In particular, 
random choices for different edges may be dependent, 
edges may cross the cut $(L,R)$ or lie inside clusters.
The set of random edges may be sampled according to a distribution that is very complex and unknown to us. 
For example, it may be sampled using the preferential attachment model. 
It can  
contain fairly large bicliques and dense structures that are found in many real-world networks~\cite{Kumar99trawlingthe,Newman2006}.

\begin{definition} Consider a set of vertices $V$ and a partition of $V$ into two sets of equal size: $V = L \cup R$.
Let $\Pi_{LR}$ be the set of permutations of $V$ such that $\pi(L) = L$ and $\pi(R) = R$. 
We say that a a probability distribution $\cal D$ on $\{E \subset V\times V\}$ is permutation-invariant 
if for every permutation $\pi\in \Pi_{LR}$ and every set $E \subset V\times V$, 
we have $\Pr_{\mathcal{D}} (\pi E) = \Pr_{\mathcal{D}} (E)$.
\end{definition}

Informally, a distribution is permutation-invariant if it ``ignores'' the ``identities'' (labels) of individual vertices;
for each vertex $u$, the distribution just ``knows'' whether $u$ is in $L$ or in $R$.

\begin{definition}[Formal Definition of the Model]
Let $V$ be a set of vertices and $V = L \cup R$ be a partition of $V$ into two sets of equal size.
Let $G = (V, E_G)$ be an arbitrary graph on $V$ in which no edge crosses cut $(L, R)$.
Let $\cal D$ be an arbitrary permutation-invariant distribution of edges. 
We define a probability distribution $\Pi(L,R, E_G, {\cal D})$ of planted graphs $F$ 
with permutation-invariant random edges (PIE) as follows.
We sample a random set of edges $E_R$ from $\cal D$ and let $F = G + E_R$.
\end{definition}
We give an alternative equivalent definition in Section~\ref{sec:two-adversaries}.
Before we state our main result, we recall the definition of the Balanced Cut problem.
\begin{definition}
A cut $(S,T)$ in a graph $G = (V,E)$ is $b$-balanced if $|S| \geq bn$ and $|T| \geq bn$ (where $b\in[0,1/2]$ is parameter).
The Balanced Cut problem is to find a $b$-balanced cut $(S,T)$ in a given graph $G$ so as to minimize the number of cut edges.
\end{definition}
We show that there is an algorithm that finds a $\Theta(1)$-balanced cut $(S,T)$ of cost $O(|E_R|) + O(n\polylog n)$ w.h.p.
This result is most interesting when the following conditions hold: (1) a constant fraction of edges in $E_R$ go from $L$ to $R$, and 
(2) the number of random edges is $\Omega(n\polylog n)$. Then, the size of the cut $(S,T)$ is at most a constant times the size of the planted cut.
That is, we obtain a constant factor approximation with respect to the size of the planted cut. 
The algorithm does not know the graph $G$, the distribution $\cal D$, and the planted cut $(L,R)$.
We now formally state out main result.
\begin{theorem}
\label{thm:main}
There is a deterministic polynomial-time algorithm that given a random graph $F$ sampled from  $\Pi(L,R, E_G, {\cal D})$
finds a $\Theta(1)$-balanced cut $(S,T)$ such that
$$|E(S,T)| = O(|E_R|) + O(n\polylog n)$$
(for arbitrary sets $L$, $R$, $E_G$, and permutation-invariant distribution $\cal D$, not known to the algorithm).
The algorithm succeeds with probability $1-o(1)$ over the choice of $F$.
\end{theorem}

\subsection{Comparison with other models}
There is an extensive literature on the random planted model~\cite{BCLS84,DF87,B88,JS,DI98,McSherry,CK01,coja} and semi-random models~\cite{FK,SVM02,MMV}. 
\begin{table*}
\centering\normalsize
\begin{tabular}{|p{2.25cm}|p{2.85cm}|p{4.9cm}|p{4.8cm}|}
\hline  \textbf{\normalsize model} \rule{0pt}{2.6ex} & \textbf{\normalsize planted graphs in $L$ and $R$} & \textbf{\normalsize random edges $E_R$}  & \textbf{\normalsize algorithm finds a balanced cut of size (w.h.p.)}  \\ 
\hline 
random 
planted model~\cite{BCLS84,DF87} & 
$G(n/2,p)$ graphs & 
edges between $L$ and $R$ are sampled independently  w. p. $q$ &
\rule{0pt}{2.6ex}$|E_R|$\newline
The algorithm recovers the planted cut.
\\
\hline  
semi-random model~\cite{MMV}&  
arbitrary graphs  & 
edges between $L$ and $R$ are sampled independently  w. p. $q$  \newline
the adversary may delete random edges
&\rule{0pt}{2.6ex}$O(qn^2)$, equals $O(|E_R|)$ if the adversary does not delete edges \newline
It is impossible to find the planted cut (information-theoretically). 
 \\ 
\hline  
our model
& 
arbitrary graphs
&  
sample $E_R$ from an arbitrary permutation-invariant distribution (unknown to the algorithm)
&  
\rule{0pt}{2.6ex}$O(|E_R|)$ \newline
It is impossible to find the planted cut (information-theoretically). \\
\hline 
\end{tabular} 
\label{table:compare}
\caption{\normalfont This table compares the random planted model~\cite{BCLS84,DF87}, semi-random model~\cite{MMV}, and model proposed in this paper.
Algorithms for all three models succeed with high probability. In this table, we assume that $(p-q)n^2 > n\polylog (n)$
in the first model, and $|E_R| > n \polylog(n)$ in the second and third models.
}
\end{table*}
We compare our model with the random planted model~of Bui, Chaudhuri, Leighton and Sipser~\cite{BCLS84}
and Dyer and Frieze~\cite{DF87} and semi-random model from our previous work~\cite{MMV}
(which generalizes the model of~Feige and Kilian~\cite{FK}), see Table~\ref{table:compare}.
In the random planted model, planted graphs in $L$ and $R$ are random $G(n/2,p)$ graphs. The set of edges $E_R$ is a random subset of 
all possible edges between $L$ and $R$; every edge is present with the same probability $q < p$ (which does not depend on the edge);
all edges are chosen independently.
The semi-random model of~\cite{MMV} is significantly more general. In this model, graphs inside $L$ and $R$ are arbitrary graphs. 
However, $E_R$ is essentially the same as in the random planted model, except that we allow the adversary to delete edges between $L$ and $R$.
In the model we study in the current paper,  not only are the graphs inside $L$ and $R$ arbitrary graphs, but further, $E_R$ is sampled from an 
arbitrary permutation-invariant distribution (in particular, they can be random edges chosen with probability $q$ as in the previous models~\cite{BCLS84,DF87,MMV}). 



Bui, Chaudhuri, Leighton and Sipser~\cite{BCLS84} and Dyer and Frieze~\cite{DF87} showed how to find the planted cut w.h.p. in the random planted model (see also~\cite{B88,coja}).
This is impossible to do in our model even information-theoretically.\footnote{E.g. consider the following graph $F$:
 $F[L]$ and $F[R]$ are $G(n/2,p)$ graphs, every edge between $L$ and $R$ is present independently with probability $p$.
Then $F$ has no information about the cut $(L, R)$.} Instead, we give an approximation algorithm that gives
a constant factor approximation with respect to the size of the planted cut if conditions (1) and (2) hold.


\subsection{Motivation}
The random planted cut model (often referred to as the Stochastic Block Model) is widely used in
statistics, machine learning, and social sciences (see e.g. ~\cite{JYSX,white1976social, holland1983stochastic,fienberg1985statistical,snijders1997estimation}). 
The PIE model, which we study in this paper,  aims to generalize it,
relax its constraints, replace random choices with adversarial choices whenever possible
and yet keep the model computationally tractable.
In our opinion, the PIE model better captures real-life instances than the random planted cut model.
Consider two examples. 
The first example is clustering with noise. Suppose that we are given a set of objects $V$. The objects are
partitioned in two clusters, $L$ and $R$; but the clustering is not known to us.
We are also given a set of ``similarity'' edges $E$ on $V$. Some edges $E_G \subset E$ represent real similarities between objects in $V$;
these edges connect vertices within one cluster. 
In practice, edges in $E_G$ are not random and our model does not impose any restrictions on them
in contrast to the random planted cut model, which assumes that they are completely random.
Other edges $E_R \subset E$ are artifacts caused by measurement errors and noise.
Edges in $E_R$ are somewhat random and it is reasonable in our opinion to assume --- as we do in our model ---
that they are sampled from a permutation-invariant distribution. 
Unlike the random planted model, we do not assume that edges in $E_R$ are sampled independently.

The second example is related to social networks. 
There are many types of ties in social networks --- there are social ties between relatives, 
friends, colleagues, neighbors, people with common interests and hobbies. The whole social network
can be thought of as a superposition of separate networks with different types of ties.
It is reasonable to assume that these networks are to large extent independent; e.g., you cannot tell much about somebody's neighbors,
if you just know his or her coauthors.  

Consider a social network with several types of ties. Represent it as a graph: the vertices represent people,
and edges represent social ties. Assume that people in the social network live in different
geographical regions, cities, countries, etc. We divide all regions into two groups and denote the set
of people who live in the regions in the first and second groups by $L$ and $R$, respectively. 
Some types of ties are usually ``local'' --- they are ties between people living in the same region; 
e.g. typically friends live in the same region. Other ties are not necessarily local; e.g. 
coauthors, college classmates, and Twitter followers do not necessarily live in the same region.
%
Let $E_G$ be the set of  edges representing local ties and $E_R$ be edges representing other ties. 
Then the whole social network is the union of $E_G$ and $E_R$.
The assumption that social ties of different types are independent is formalized in our model by the condition that
$E_R$ is sampled from a permutation-invariant distribution.
That is, we take two social networks $G=(V_G, E_G)$ and $H=(V_H, E_R)$, choose a random correspondence between vertices of $G$ and $H$, and
then identify corresponding vertices (using the notation, which we introduce in the next section, we consider the graph $F = G \boxplus_\pi H$
for a random permutation $\pi$).


We believe that techniques similar to those we present in the current paper can be applied to other graph partitioning 
and combinatorial optimization problems. We hope that these techniques will be useful for solving real world problems 
on networks that we encounter in practice.

\subsection{Model with Two Adversaries}\label{sec:two-adversaries}
We use an alternative equivalent formulation of our model in the rest of the paper.
Let $G = (V_G, E_G)$ and $H=(V_H, E_H)$ be two graphs on $n$ vertices, and $\pi:V_H\to V_G$ be a bijection. 
Define the graph $F = G \boxplus_\pi H$ on $V_G$ by $E_F = E_G \cup \pi(E_H)$.
Let $V_G = L_G\cup R_G$ and $V_H=L_H\cup R_H$ be partitions of $V_G$ and $V_H$ into sets of size $n/2$.
Define  
$\Pi_{LR}\equiv \{\pi:V_H\to V_G \;:\;\pi(L_H)=L_G\text{ and } \pi(R_H)=R_G\}$
to be the set of all bijections mapping $L_H$ to $L_G$ and $R_H$ to $R_G$.

Suppose now that one adversary chooses an arbitrary graph $G$ with no edges between $L_G$ and $R_G$,
and another adversary chooses an arbitrary graph $H$ (both adversaries know the partitions $V_G = L_G\cup R_G$
and $V_H = L_H\cup R_H$). Then the nature chooses a bijection $\pi \in \Pi_{LR}$ uniformly at random. 
We obtain a graph $F = G \boxplus_{\pi} H$. 

\begin{theorem}\label{thm:mainPlusPi}
There exists a deterministic polynomial-time algorithm that given a graph $F=G\boxplus_{\pi} H$ outputs a $\Theta(1)$-balanced partition
of $V_F=V_G$ into two sets $L'$ and $R'$. If there are no edges between $L_G$ and $R_G$ in $G$, then
the cost of the cut $(L',R')$ is bounded by $O(|E_H|+n \log^3 n)$ with probability $1-o(1)$
over a random choice of $\pi\in \Pi_{LR}$.
\end{theorem}

\textbf{Remark:}
To simplify the exposition we do not attempt to optimize the constants in the $O(\cdot)$ notation. The additive term $n \log^3 n$
can  be slightly improved.

This theorem implies Theorem~\ref{thm:main}. Indeed, if $|E_R|$ is a random permutation invariant set of edges, then 
$E_R$ is distributed identically to $\pi(E_R)$, where $\pi$ is a random permutation from $\Pi_{LR}$. Thus, graphs
$(V_G,E_G\cup E_R)$ are distributed identically to graphs $V_G\boxplus_{\pi}(V_G,E_R)$. The algorithm from Theorem~\ref{thm:mainPlusPi} 
succeeds with probability $1-o(1)$ on graphs $V_G\boxplus_{\pi}(V_G,E_R)$ for every fixed $E_R$ and random $\pi\in \Pi_{LR}$.
Thus, it succeeds with probability $1-o(1)$ on graphs $(V_G,E_G\cup E_R)$.

\subsection{Techniques}
We present a very high-level overview of the algorithm.
We are given a graph $F = G \boxplus_{\pi} H$ and our goal is to find a balanced cut of size roughly
$O(|E_H|)$. 
We assume that $|E_G| \gg |E_H|$ as otherwise any balanced cut cuts $O(|E_H|)$ edges and we are done.
We write an SDP relaxation for Balanced Cut. The  relaxation is 
similar but slightly different from the one of Arora, Rao and Vazirani~\cite{ARV}
(see Section~\ref{sec:prelim} for details). 
The SDP solution assigns a vector $\varphi(u)$ to every vertex $u \in V_G$. 
The objective function is to minimize $\sum_{(u,v)\in E_F}\|\varphi(u) - \varphi(v)\|^2$.
The SDP constraints ensure that all vectors lie on a sphere ${\cal S}$ of radius $\sqrt{2}/2$.
Given an SDP solution, we say that an edge $(u,v)$ is $\delta$-short if $\|\varphi(u)  - \varphi(v)\|^2 \leq \delta$, where $\delta$ is a fixed constant,
and that it is $\delta$-long, otherwise.

For the sake of discussion, let us first make a very unrealistic assumption
that the SDP solution is determined by the set of edges $E_G$ and does not depend on the set 
of random edges $E_R=\pi E_H$.
Assume furthermore that all vectors $\{\varphi(u)\}$ are distributed more-or-less uniformly on the sphere $\cal S$;
more precisely, assume that every ball of radius $\delta$ w.r.t. the squared Euclidean distance contains very few vectors $\varphi(u)$.
Then for every edge $e=(u,v)\in E_H$, the probability over $\pi$ that vectors $\varphi(\pi u)$ and $\varphi(\pi v)$
lie in the same ball of radius $\delta$ is very small, and thus  $\pi e$ is a long edge with high probability.
Now the total number of long edges in $F$ is at most $|E_R|/\delta$ since each long edge contributes
at least $\delta$ to the SDP objective function and the SDP value is at most the cost of the planted cut. 
This discussion suggests an approach to the problem.
Let us remove all long edges in $F$. When we do so, we decrease the number of edges in $E_R$ by a constant factor and
cut only a constant number of edges in $E_G$ for each cut edge in $E_R$. 
We repeat this step over and over until (almost) all random edges are cut. The total number of removed edges does
not exceed $O(|E_R|)$, as required.

There are several problems with this argument. 
\begin{enumerate} 
\item The SDP solution does depend on the set $E_R$.  
\item Vectors $\varphi(u)$ are not uniformly distributed on the sphere $\cal S$, in general. In fact, there are only two possible values
for vectors $\varphi(u)$ in the intended integral solution.
\item We will not make any progress if we just run the same procedure over and over.
\end{enumerate}

We use a \textit{Heavy Vertices Removal} procedure to deal with the second and third problems. Conceptually, the procedure finds balls of radius $\delta$ that contain many vertices and cuts them off from $F$ so that the total number of cut edges is small. 
We apply this procedure more-and-more aggressively in consequent iterations.

The first problem is much more serious and most of this paper describes how to solve it.
Recall that we assume that $|E_G| \gg |E_H|$ and thus most edges in $E_G$ are short. That means informally that short edges of $G$ form a ``skeleton'' of $G$ --- edges in this skeleton are short and they locally constrain how the SDP solution may look like. 
The skeleton does not necessarily cover the whole graph $G$;
moreover, even if initially the skeleton covered the whole graph $G$, it may no longer cover $G$ after we perform a few iterations of the algorithm. We use a special \textit{Damage Control}
procedure to remove vertices not covered by the skeleton. This is a tricky step since the algorithm does not know which edges are in $E_F$ and which are in $E_R$ and consequently cannot compute the skeleton. 

Now to make our argument work, we need to show that few edges in $E_R$ are short (and thus many edges in $E_R$ are long). Assume to the contrary that many edges in $E_R$ are short.
Then we can also find a skeleton in the graph $H$.
We prove in the Main Structural Theorem that if both graphs $G$ and $H$ have skeletons then there is a very efficient encoding of $\pi$; namely, we prove that the prefix Kolmogorov complexity $KP(\pi)$ of $\pi$ is much less than $\log_2 |\Pi_{LR}|$. 
The encoding consists of two parts. 
We identify two relatively small sets of vertices $Q_G \subset V_G$ and $Q_H \subset V_H$ and record values of $\varphi(u)$ for $u\in Q_G$ and values of $\varphi(\pi(x))$ for $x \in Q_H$ in the first part of the encoding. The first part of the encoding
allows us to approximately reconstruct values of 
$\varphi(u)$ for all vertices $u\in V_G$ and values of $\varphi(\pi(x))$ for all vertices 
$x\in V_H$ using that edges in the skeletons for $G$ and $H$ are short. 
Note that if $u = \pi(x)$ then $\varphi(u) = \varphi(\pi(x))$. Thus if we knew the
values of $\varphi(u)$ and $\varphi(\pi(x))$ exactly and all values $\varphi(u)$ were distinct, we would be able to reconstruct $\pi$: $\pi(x) = \varphi^{-1} (\varphi(\pi(x)))$. In fact, the encoding gives us only
approximate values of $\varphi(u)$ and $\varphi(\pi(x))$ but still it tells
us that $\pi(x)$ is equal to such $u$ that $\varphi(u)$ and $\varphi(\pi(x))$
are close. Given that, we can very efficiently record additional information necessary
to reconstruct $\pi$ in the second part of the encoding. We show that the 
total length of the encoding is much less than $\log_2 |\Pi_{LR}|$ bits and thus
$KP(\pi) \ll \log_2 |\Pi_{LR}|$.

Since an exponentially small fraction of permutations in $\Pi_{LR}$ has prefix Kolmogorov
complexity much smaller than $\log_2 |\Pi_{LR}|$, the probability that  
both graphs $G$ and $H$ have skeletons is exponentially small and thus $E_R$ contains 
many short edges with high probability.

We note that the algorithm is quite involved and technical, and 
we cannot describe it accurately in the introduction. Thus the overview given above is very informal.
It only gives a rough idea of how our algorithm and analysis work. In particular, we do not use the informal 
notion of ``skeleton'' in the paper.


\paragraph{Technical Comparison}
We use ideas introduced in papers on semi-random instances of Unique Games~\cite{KMM} and on 
semi-random instances of graph partitioning problems~\cite{MMV}.
The very high-level approach of this paper is somewhat similar to that of our previous work~\cite{MMV}.
As in~\cite{MMV},  our algorithm iteratively removes long edges and uses a Heavy Vertices Removal procedure.
However, overall the algorithm and analysis in this paper are very different from that of~\cite{MMV}.
In~\cite{MMV}, the proof of the main structural theorem relies on the fact that $H$ is a random $G(n/2,n/2,q)$ bipartite graph.
That ensures that  most edges in $E_R$ are long no matter what the graph $G$ is. However, 
that is no longer the case in the present paper: The graph $(V,E_R)$ can be a completely
arbitrary graph. It does not have to be an expander or ``geometric expander'' (the notion we used in~\cite{MMV}). To prove the structural theorem, we have to analyze the skeleton formed 
by edges in $E_G$. As a result, the proof of the structural theorem is completely different from the proof in~\cite{MMV}.
The algorithm is also significantly different. It needs to perform an extra Damage Control step
and the Heavy Vertices Removal Step is quite different from that in~\cite{MMV}. 
There are numerous other differences between
algorithms.

\section{Preliminaries}\label{sec:prelim}
We work with the model described in Section~\ref{sec:two-adversaries}. We denote the number of vertices in $F=G\boxplus_{\pi}H$ by $n$ and 
let 
$$d = \max\{2|E_H| /n, C\log^3 n\}$$
 for sufficiently large constant $C$ 
($d$ equals the average degree of vertices in the graph $H$ if the average degree is greater than
$C\log^3 n$).
We assume without loss of generality that $d$ is known to the algorithm (the algorithm 
can find $d$ using binary search).  We denote the degree of a vertex $u$ in $F$ by $\deg(u, F)$, in $G$ by $\deg(u,G)$, and in $H$ 
by $\deg(u,H)$.

Our algorithm performs many iterations; in each iteration, it solves an SDP relaxation for Balanced Cut on a subgraph $F'$ of $F$.
The relaxation for $F'$ assigns a vector $\varphi(u) \in \mathbb{R}^n$ to every vertex $u$ of $F'$. The SDP is shown in Figure~\ref{fig:SDP}.
\begin{figure}
\OpenFrame
\begin{align}
&\text{minimize:\ } \sum_{(u,v)\in E_{F'}} \|\varphi(u) - \varphi(v)\|^2 
 \label{SDP:objective}\\ 
\intertext{such that for every $u,v,w\in V_{F'}$,}
&\hspace*{2.2mm}\|\varphi(u)\|^2 = \frac{1}{2}  \label{SDP:normalization}\\ 
&\hspace*{2.2mm}\sum_{v\in V_{F'}} \left(1 - \| \varphi(u) -  \varphi(v) \|^2 \right) \leq  n/2
\label{SDP:spreading}\\ 
&\hspace*{2.2mm}\|\varphi(u)- \varphi(v)\|^2 + \|\varphi(v)- \varphi(w)\|^2 \geq \|\varphi(u)- \varphi(w)\|^2. \label{SDP:triangle}
\end{align}
\CloseFrame
\caption{SDP relaxation for Balanced Cut}
\label{fig:SDP}
\end{figure}
The intended integral solution is $\varphi(u) = e_1/\sqrt{2}$ if $u\in L$ and $\varphi(u) = e_2/\sqrt{2}$ if $u\in R$, where $e_1$ and $e_2$ are 
two fixed orthogonal unit vectors. The intended solution satisfies all SDP constraints.
We denote the cost of a feasible SDP solution $\varphi$ for a graph $F'$ by $\sdpcost(\varphi, F')$:
$$\sdpcost(\varphi, F') = \sum_{(u,v) \in E_{F'}} \|\varphi(u) - \varphi(v)\|^2.$$
The cost of the intended SDP solution  
equals the number of edges from $L$ to $R$. Since only random edges in $F$ go from $L$ to $R$, it is at most $|E_R|$. 
Note that the optimal SDP solution $\varphi_{\mathrm{opt}}$ for $F$ costs at most as much as the intended solution; thus 
$\sdpcost(\varphi_{\mathrm{opt}},F) \leq |E_R| \leq dn/2$.  

Our SDP relaxation for Balanced Cut is slightly different from that of Arora, Rao and Vazirani: we use different normalization in~(\ref{SDP:normalization})
and use different spreading constraints~(\ref{SDP:spreading}). However, the algorithm of Arora, Rao and Vazirani works with our SDP.
We denote the approximation factor of the algorithm 
by $D_{ARV}=O(\sqrt{\log n})$. 
The algorithm given an SDP solution $\varphi$ for a subgraph $F'$ of $F$ finds a cut $(L',R')$
that cuts at most $D_{ARV} \, \sdpcost(\varphi,F')$ edges  such that both sets $L'$ and $R'$ contain at most $cn$ vertices for some absolute constant $c_{ARV}\in (0,1)$.
Let $T=\roundup{\log_2 D_{ARV}}=O(\log \log n)$.

We say that an edge $(u,v)$ is $\nicefrac{\delta}{2}$-short if $\|\varphi_t(u)-\varphi_t(v)\|^2\leq\nicefrac{\delta}{2}$; otherwise,
it is $\nicefrac{\delta}{2}$-long.
In our algorithm, we use five parameters $K$, $\beta = 200K$, $\alpha = 50K$,  $\delta=1/12$ and 
$D_n = \max\{D_{ARV}, \alpha\}$. The parameter $K$ is a sufficiently large constant.
Let $\VGleqA = \set{u\in V_F: \deg(u,G) < \alpha d}$. 
It will be convenient for us to assume that $|\VGleqA| \leq n/\alpha$. If this is not the case, we run a very simple algorithm for Balanced Cut, which we present
in Appendix (see Lemma~\ref{lem:simpleBalancedCut}).

Our algorithm iteratively cuts edges and removes some components of the graph (a component is an
arbitrary subset of vertices).
We say that a vertex is \textit{removed} if it lies in a removed component; otherwise, we say that the vertex is \textit{active}.
We distinguish between cut and removed edges. An edge $e$ is \textit{cut} if the algorithm cuts it, 
or if $e$ belongs to the edge boundary of a removed component. An edge is \textit{removed} if either it is cut or at least one of its endpoints is removed.

The algorithm we present partitions the graph into several pieces and cuts at most $O(dn) = O(|E_H| + n\log^3 n)$
edges. The size of each piece is at most $\max(c_{ARV}, \nicefrac{3}{4}) \, n$. We can combine all pieces into two 
$\max(c_{ARV}, \nicefrac{3}{4})\,n$-balanced parts. The number of edges between these parts is at most $O(dn)$
as required in Theorem~\ref{thm:mainPlusPi}.

\section{Algorithm} \label{sec:algo}
We now present the algorithm. The main steps of the algorithm are given in Figure~\ref{fig:mainAlgorithm}. Below we describe the algorithm
in more detail.

\smallskip\noindent\textbf{Budget allocation:} We store a budget for every vertex $u$. We use this budget to keep track of 
the number of cut edges incident on $u$. We do that to identify vertices we need to remove at Steps 3 and 4, and also to bound the total number of cut
edges. Initially, the algorithm assigns
a budget to every vertex $u$: vertex $u$ gets a budget of 
$\beta d$ if $\deg(u,F)\geq \alpha d$;
and a budget of $\alpha d$ if $\deg(u,F) < \alpha d$. We denote the budget of a vertex $u$ by $\budget (u)$ and the budget of
a set $S$ by $\budget (S)\equiv \sum_{u\in S} \budget (u)$. We allocate an extra budget of $3nd/\delta$ units. We keep this 
extra budget in the variable $\extraBudget$. 

\smallskip\noindent\textbf{Main loop:} The algorithm works in $T$ iterations. We let $F_1(0)$ to be the original graph $F$.  
Consider iteration $t$. At Step 1, the algorithm solves the SDP relaxation for the graph $F_1(t)$
and obtains an SDP solution $\varphi_t:V_{F_1(t)}\to \bbR^n$, which is a mapping of vertices of the graph $F_1(t)$ to 
$\bbR^n$. At Step 2, the algorithms cuts all $\nicefrac{\delta}{2}$-long edges i.e., edges $(u,v)$ such that 
$\|\varphi_t(u)-\varphi_t(v)\|^2\geq\nicefrac{\delta}{2}$. At Step 3, the algorithm runs the Heavy Vertices Removal procedure
and at Step 3, the algorithm runs the Damage Control procedure. We describe the details of these three steps in Sections~\ref{sec:LongEdgesRemoval}, \ref{sec:HeavyVerticesRemoval} and~\ref{sec:DamageControl}. The Heavy Vertices Removal and Damage Control procedures remove some 
vertices from the graph.  Edges on the boundary of the components removed by these procedures 
at iteration $t$ are cut. We denote them by $\Upsilon_3(t)$ and $\Upsilon_4(t)$, respectively.
We denote the set of long edges cut at Step 2 by $\Upsilon_2(t)$. 
Finally, we denote the graphs obtained after Steps 2, 3, 4
by $F_2(t)$, $F_3(t)$ and $F_4(t)$. 
At iteration $t$, 
after completion of Step~$i$, the set of active vertices is $V_{F_i(t)}$.

\smallskip\noindent\textbf{Budget updates:} When we cut a long edge $(u,v)$ at Step 2, we increase the budget of vertices $u$, $v$ by 1 and
decrease the extra budget by 3. When we cut
an edge $(u,v)$ at Step 3 or Step 4, we increase the budget of the active endpoint (the one we do not remove) by 1. 
Thus, we have the following invariant: \textit{The budget of every active vertex $u$ always equals the initial budget of $u$ plus the number of cut edges incident on $u$ in the graph $F$.}

\smallskip\noindent\textbf{Final partitioning:} After the last iteration of the loop is completed, we partition the graph 
$F_1(T) = F_4(T-1)$ into two balanced pieces $L'$ and $R'$ using the algorithm of Arora, Rao and Vazirani.
We output $L'$, $R'$ and all components removed at Steps 3 and 4 (in all iterations).

\begin{figure}[t]
\OpenFrame

\noindent \textbf{Main Algorithm} 

\medskip
\noindent \textbf{Input: }a graph $F=G\boxplus_{\pi}H$ (graphs $G$, $H$, and the permutation $\pi$ are hidden from the algorithm).

\noindent \textbf{Output: } a partitioning of $F$ into pieces of size at most $c n$ for some $c< 1$.

\begin{itemize}
\item \textbf{Set the parameters:} $\beta = 200K$, $\alpha=50\beta$, $\eta_t = 2^{-t}$ (for $t\in \bbZ^+$). Let
$D_{ARV}=O(\sqrt{\log n})$ be the approximation ratio of the ARV algorithm; $D_n = \max\{D_{ARV}, \alpha\}$; $T=\roundup{\log_2 D_{ARV}}$.
\item \textbf{Allocate budget:} For every vertex $u\in U$, set $\budget(u)=\beta d$ if $\deg (u, F)\leq \alpha d$;
and $\budget(u)=\alpha d$ if $\deg (u, F)\geq \alpha d$.
\item \textbf{Let } $F_1(0)=F$.
\item \textbf{for $t=0$ to $T-1$ do}:
\begin{enumerate}
\item Solve the SDP on the graph $F_1(t-1)$. Denote the SDP solution by $\varphi_t:V\to \bbR^n$.
\item Remove $\nicefrac{\delta}{2}$-long edges. Update the budgets.
\item Run Heavy Vertices Removal procedure with $\eta_t = 2^{-t}$. Update the budgets.
\item Run Damage Control procedure. Update the budgets.
\item Denote the graphs obtained after Steps 2--4 by $F_2(t)$, $F_3(t)$ and $F_4(t)$.
Denote the set of edges cut at these steps by $\Upsilon_2(t)$, $\Upsilon_3(t)$ and $\Upsilon_4(t)$. 
Let $F_1(t+1)=F_4(t)$.
\end{enumerate}
\item Partition the graph $F_1(T)$ into two graphs $L'$ and $R'$ using the ARV algorithm.
\item \textbf{Return} $L'$, $R'$ and all components removed at Steps 3 and 4.
\end{itemize}
\CloseFigureFrame
\caption{Main steps of the algorithm. We present the algorithm in more detail below.}
\label{fig:mainAlgorithm}
\end{figure}

\subsection{Analysis}

We show that the algorithm returns a solution of cost at most $O(|E_H|)$ if the graph $F$ satisfies Structural Properties 1--4, 
which we describe in  Section~\ref{sec:StructPropDefs}. Then we show that the graph $F=G\boxplus_{\pi}H$ satisfies these properties with high probability (i.e., with probability $(1-o(1))$). 

Define the total budget after Step $i$ at iteration $t$ to be the sum of budgets of active vertices plus the extra budget:
$$\totalBudget = \sum_{u \text{ is active}} \budget (u) + \extraBudget.$$
We prove that at every step of the algorithm the total budget does not increase (though the budgets of some vertices
do increase). Furthermore, we show that whenever we cut a set of edges $\Upsilon_i(t)$, the total budget decreases
by at least $|\Upsilon_i(t)|$. In other words, we pay a unit of the budget for every cut edge.

\begin{lemma}\label{lem:mainBudget}
Let $b_{before}$ be the total budget before executing Step $i$ at iteration $t$; and 
let $b_{after}$ be the total budget after executing Step $i$ at iteration $t$. If $F=G\boxplus_{\pi}H$ satisfies Structural Properties 1--4, then 
$$b_{after} \leq b_{before} - |\Upsilon_i(t)|.$$
\end{lemma}
At Steps 1 and 5, we neither update the budgets of vertices, nor do we change the set of active vertices, so the total budget does not change. 
We consider Steps 2--4 in Lemmas~\ref{lem:LongEdgesBudget}, \ref{lem:HeavyVerticesBudget}, and~\ref{lem:DamageControlBudget}. 
In Lemma~\ref{lem:LongEdgesBudget}, we also show that the extra budget and hence the total budget is always non-negative (the budgets
of vertices may only increase, but the extra budget may only decrease). 

Structural Property 3 (see Section~\ref{sec:StructPropDefs}) guarantees that the total budget initially allocated by the algorithm is
at most $\nicefrac{3}{2}\,\beta d n$. Hence, the total number of edges cut by the algorithm is at 
most $\nicefrac{3}{2}\,\beta d n$. We denote the set of all cut edges by $\Upsilon$:
$$\Upsilon = \bigcup_{\substack{i\in\{2,3,4\}\\t\in\{0,\dots, T-1\}}} \Upsilon_{i}(t).$$ 

The algorithm of Arora, Rao and Vazirani partitions the graph $F_4(T)$ into two pieces of size at most $cn$ each (where $c < 1$ is an absolute constant).
In Sections~\ref{sec:HeavyVerticesRemoval} and \ref{sec:DamageControl}, we show that each component removed at Steps 3 and 4 has size at most
$\nicefrac{3}{4}\,n$ (see Lemma~\ref{lem:HeavyVerticesBudget} and Lemma~\ref{lem:DamageControlBudget}). Hence, all pieces in the
partition returned by the algorithm have size at most $\max(\nicefrac{3}{4},c) n$.

Now we need to verify that the size of the cut separating different pieces in the partition is at most $O(dn)$. This cut
contains edges from $\Upsilon$ and edges cut by the ARV algorithm. We already know that $|\Upsilon|\leq \frac{3}{2}\beta dn 
=O(dn)$. It remains to prove that the ARV algorithm cuts $O(dn)$ edges. The proof follows from
Theorem~\ref{thm:mainSDPbound}, which is central to our analysis.

\begin{theorem}\label{thm:mainSDPbound} 
If the graph $F=G\boxplus_{\pi}H$ satisfies Structural Properties 1--4, then for every $t\in\{0,\dots, T\}$,
$$\sdpcost(\varphi_t, F_1(t)) \leq 8 K \eta_t dn,$$
where $\varphi_t$ is the optimal SDP solution for $F_1(t)$, $\eta_t\equiv 2^{-t}$, and
$K$ is an absolute constant.
\end{theorem}

We also use this theorem to prove Lemma~\ref{lem:HeavyVerticesBudget}, which bounds the number of edges cut by the Heavy Vertices removal
procedure. For~$T=\roundup{\log_2 D_{ARV}}$, we get that $\sdpcost(\varphi_T, F_1(T)) \leq K dn/D_{ARV}$. The algorithm
of Arora, Rao and Vazirani outputs an integral solution of cost at most 
$$D_{ARV}\times \sdpcost(\varphi_T, F_1(T)) \leq D_{ARV} \times \frac{K dn}{D_{ARV}} = Kdn.$$
That is, the size of the cut between $L'$ and $R'$ is at most $Kdn$.
This finishes the analysis of the algorithm.

\subsection{Notation}
Before proceeding to the technical part of the analysis, we set up some notation. During the execution of the algorithm,
we remove some vertices and cut some edges from the graph $F$. For the purpose of analysis, 
we will shadow these removals in the graphs $G$ and $H$. For every $F_i(t)$ we define 
two graphs $G_i(t)$ and $H_i(t)$. The vertices of these graphs are the vertices of $F_i(t)$.
The edges of $G_i(t)$ are edges of $F_i(t)$ that originally came from $G$; 
the edges of $H_i(t)$ are edges of $F_i(t)$ that originally came from $H$. Note that
$G_1(0)$ equals $G$, $H_1(0)$ is isomorphic to $H$, and the isomorphism between 
$H$ and $H_1(0)$ equals $\pi$.   

We denote by $\deg(u,F_i(t))$, $\deg(u,G_i(t))$, $\deg(u,H_i(t))$ the degree of the vertex $u$ in the graph $F_i(t)$,
$G_i(t)$ and $H_i(t)$, respectively. We denote by $\deg(u,F)$, $\deg(u,G)$, $\deg(u,H)$ the degree of $u$ in 
the original graphs $F$, $G$, $H$. Note that strictly speaking $\deg(u,H)$ is the degree of the vertex $\pi^{-1}(u)$
in the graph $H$.

Given a graph $G$, an SDP solution $\varphi: V_G\to \bbR^n$, and a positive number
$\delta>0$, we denote by $\short_{\varphi,\delta} (u,G)$ and $\short_{\varphi,\delta} (u,H)$ 
the number of $\delta$-short edges w.r.t the SDP solution $\varphi$ leaving vertex $u$ in $G$ and $\pi H$,
respectively. Finally, we denote by $N_F(u)$, $N_G(u)$ the set of neighbors 
of $u\in V_G$ in the graphs $F$ and $G$ and by $N_H(x)$ the set of neighbors of $x\in V_H$ in the graph $H$.

\subsection{Overview of the Proof}
The analysis of the algorithm relies on Theorem~\ref{thm:mainSDPbound}. It states that the cost of the optimal SDP solution for $F_4(t) = F_1(t+1)$ is $O(dn/2^t)$. 
To prove this theorem, we construct an SDP solution of cost $O(dn/2^t)$.
To this end, we first divide the graph $F_4(t)$ into two sets, the set of ``undamaged'' vertices $W$ and the set of ``damaged'' vertices  $\bar W$. Then we further subdivide $W$ into $W\cap L$ and $W\cap R$ and get a partition of $F_4(t)$ into three pieces $W\cap L$, $W \cap R$, and $\bar W$. We prove that each piece contains at most $n/2$ vertices and the total number of edges cut by the partition is $O(dn/2^t)$ (we outline the proof below). The partition defines a feasible integral SDP solution that assigns the same vector to vertices in one part and orthogonal vectors to vertices in different parts. The cost of this SDP solution is $O(dn/2^t)$ as required.

Thus we need to prove that the partition into $W\cap L$, $W\cap R$ and $\bar W$
is balanced and cuts few edges. We first deal with the part $\bar W$.
We run the Damage Control procedure that cuts off some components of the graph so as to ensure that $|\bar W| \leq n/2$ and more importantly $|\partial \bar W| \leq O(dn/2^t)$.
We describe the procedure and prove that it cuts a small number of edges if the graph satisfies Structural Properties 2--4 in Section~\ref{sec:DamageControl}; we show that a graph in the PIE model satisfies these properties w.h.p. in Sections~\ref{sec:StructProp23} and~\ref{sec:StructProp4}.

Now consider parts $W\cap L$ and $W \cap R$. We immediately have that $|W\cap L| \leq |L| = n/2$ and $|W\cap R| \leq |R| = n/2$. There are no edges between $W\cap L$ and $W\cap R$ in $G_4(t)$ (since $(L,R)$ is the planted cut). It remains to show that
there are at most $O(dn/2^t)$ edges between $W\cap L$ and $W\cap R$ in $H_4(t)$.
Note that all edges in $H_4(t)$ are $\nicefrac{\delta}{2}$-short w.r.t. $\varphi_t$
since we cut all  $\nicefrac{\delta}{2}$-long edges at Step 1.
We prove in the Main Structural Theorem (Theorem~\ref{thm:MainStructuralTheorem}) that there are at most $O(dn/2^t)$ $\nicefrac{\delta}{2}$-short edges in the induced graph $H_4(t)[W]$ and thus there are at most $O(dn/2^t)$ edges between $W\cap L$ and $W\cap R$ in $H_4(t)$.

We now sketch the proof of the Main Structural Theorem (Theorem~\ref{thm:MainStructuralTheorem}). 
We present the proof in a simplified setting; most steps are somewhat different in the actual proof.
We assume that all vertices in $H$ have degree $d$.
Denote $\eta = 1/2^t$. All vertices in $W$ satisfy several properties --- if a vertex does not satisfy these properties it is removed either by the Heavy Vertices Removal or Damage Control procedure.
The Heavy Vertices Removal procedure removes all vertices $u$ such that
the ball $\{v:\|\varphi_t(u)-\varphi_t(v)\|^2 \leq 3\delta\}$ has a budget of $\eta \beta dn$. 
We show that this implies that for every active $u$ there are at most $2\eta n$ vertices with 
more than $\beta d/2$ neighbors in the ball of radius $2\delta$ around $u$ (see Lemma~\ref{lem:prop2isSatisfied}).
The Damage Control procedure removes all ``damaged'' vertices.
We do not describe the Damage Control procedure in this overview, but we note that in particular it guarantees that 
$\short_{\varphi_t, \nicefrac{\delta}{2}} (u, G) \geq \beta d$ for all vertices
$u \in W$. 

For simplicity, we will assume now that $W = V_G$.
Recall that $F = G \boxplus_\pi H$ in our model.
We show that if $H_4(t)$ contains more than $K \eta dn$ $\nicefrac{\delta}{2}$-short edges then there is a binary encoding of $\pi$ with much fewer than $\log_2 |\Pi_{LR}|$ bits.
Since any encoding needs $\log_2 |\Pi_{LR}|$ bits to encode a typical permutation in $\Pi_{LR}$, the probability that for a random $\pi \in \Pi_{LR}$
the graph $H_4(t)$ contains more than $K\eta dn$ short edges is very small.

We fix a permutation $\pi$ and assume to the contrary that $H_4(t)$ contains more than $K \eta dn$ $\nicefrac{\delta}{2}$-short edges. We are going to show that there is a short encoding of $\pi$. 
We sample two random subsets $Q_G\subset V_G$ and $Q_H \subset V_H$. Each vertex
of $G$ and $H$ belongs to $Q_G$ and $Q_H$ (respectively) with probability $q= D_n/d$. Additionally, we choose random orderings of $Q_G$ and $Q_H$. 
Note that $Q_G$ and $Q_H$ are of size approximately $qn$.
From now on all random events that we consider are with respect to our random choices
of $Q_G$, $Q_H$ and their ordering (not the random choice of $\pi$).

For every vertex $x\in V_H$, let $x'$ be the first neighbor of $x$ in $Q_H$ w.r.t. to the random ordering of $Q_H$ if it exists.
Note that the probability that $x'$ is defined for a given $x\in V_H$ is $1-(1-q)^d\approx 1 - e^{-D_n}$; that is, $x'$ is defined for most vertices $x$.
Vertex $x'$ is uniformly distributed in $N_H(x)$. Thus the edge $(x,x')$ is short 
with probability $\short_{\varphi_t, \nicefrac{\delta}{2}} (u, H)/d$. The expected number of vertices $x$ such that $(x,x')$ is short is 
$$\sum_{u\in V_H} \short_{\varphi_t, \nicefrac{\delta}{2}} (u, H)/d \geq K \eta nd /d = K\eta n.$$

If $x'$ exists and $(x,x')$ is short, define
$
B     = \{v: \|\varphi_t(v)-\varphi_t(\pi(x'))\|^2 \leq \delta\} $ and $
\Xi(x)=\{v: |Q_G\cap N_G(v)\cap B| \geq q \beta d \}
$.
Recall that for every ball of radius $2\delta$ (or less), there are at most 
$2\eta n$ vertices with more than $\beta d/2$ neighbors in the ball. Thus,
$|\Xi(x)|\lesssim 2\eta n$. Now note that $\short_{\varphi_t, \nicefrac{\delta}{2}} (\pi(x), G) \geq \beta d$ thus there are at least $\beta d$ vertices in $N_G(\pi(x))$ at distance at most
$\nicefrac{\delta}{2} + \|\varphi_t(\pi(x)) - \varphi_t(\pi(x'))\|^2 \leq \delta$ from $\pi(x')$. That is,
$|N_G(\pi(x)) \cap B| \geq  \beta d$ and in expectation
$Q_G\cap N_G(\pi(x)) \cap B$ contains at least $q \beta d$ vertices. Therefore, 
$\pi(x) \in \Xi(x)$ w.h.p. 

Let $\cal X$ be the set of vertices $x$ such that $x'$ exists, the edge $(x,x')$ is 
$\nicefrac{\delta}{2}$-short, $\pi(x) \in \Xi(x)$ and
$|\Xi(x)| \leq 2\eta n$. As we showed above, $\cal X$ contains approximately $K\eta n$ vertices.
We are now ready to explain how we encode the permutation $\pi$.
We first record sets $Q_G$, $Q_H$ and orderings of $Q_G$ and $Q_H$ in our encoding.
For each $u\in Q_G$ we record $\varphi_t(u)$; for each $x\in Q_H$ we record 
$\varphi_t(\pi(x))$. We record the set $\calX$ and the restriction of $\pi$ to the complement of $\bar \calX$. 
Finally, for each $x\in \calX$, we record the sequential number of $\pi(x)$ in the set $\Xi(x)$ w.r.t. an arbitrary fixed ordering of $V_G$
(i.e. the number of elements preceding $\pi(x)$ in $\Xi(x)$).

We show how to decode $\Pi_{LR}$ given our encoding of $\pi$.
We know the value of $\pi(x)$ for $x\in \bar \calX$, so consider $x\in \calX$. 
First compute $x'$ and $\Xi(x)$.
The encoding contains all the necessary information to do so.
Now find $\pi(x)$ in $\Xi(x)$ by its sequential number in $\Xi(x)$. We showed that $\pi$ is determined by its encoding.

Now we estimate the length of the encoding. Sets $Q_G$ and $Q_H$ are of size 
approximately $qn$. We need 
$O(q n \log (1/q))$ bits to record them, $O(q n \log (qn))$ bits to record their orderings, $O(q n \log n)$ bits to record vectors $\{\varphi(u)\}_{u \in Q_G}$ and $\{\varphi(\pi(x))\}_{x \in Q_H}$ with the desired precision (that follows from the Johnson---Lindenstrauss lemma). We need $|\calX|\bLog (1/(\eta K))$ bits to record $\calX$ (since
the size of $|\calX|$ is approximately $K\eta n$).
We need at most $\log_2 ((n/2)! (n/2 - {\calX})!)$ bits to record the restriction
of $\pi$ to $\bar {\cal X}$. Finally, we need $\log_2 |\Xi(x)| = \log_2 (\eta n) + O(1)$ bits for
each vertex $u \in \cal X$ to record its position in $\Xi(x)$. In total, we need
$$
\log_2 ((n/2)! (n/2 - |{\cal X}|)!)  + |{\cal X}| \log_2 (n/K) + O(qn \log n) 
$$
bits. In contrast, we need at least $\log_2 ((n/2)! (n/2)!)$ bits to encode a ``typical'' permutation in $\Pi_{LR}$ (no matter what encoding scheme we use). That is, the encoding of $\pi$ is {\em{shorter}} than the encoding of a typical permutation by at least
\begin{multline*}
\log_2 ((n/2)!(n/2)!) - \big(\log_2 ((n/2)! (n/2 - |{\cal X}|)!) +{} \\ {}+ |{\cal X}| \log_2 (n/K) +  O(qn \log n) \big) \approx {} \\ \approx
|{\cal X}| (\log_2 n - \log_2 (n/K)) -
O(qn \log n)  = {}\\ |{\cal X}| \log_2 K  - O(q n \log n)    \approx 
K\eta n \bLog K -
O((D_n/d)\, n \log n).
\end{multline*}
The expression is large when $d \gtrsim \log^3 n$. We conclude that a random permutation $\pi\in \Pi_{LR}$ does not satisfy the condition of the
Main Structural Theorem with small probability.

\subsection{Structural Properties --- Definitions}\label{sec:StructPropDefs}

We now describe the Structural Properties that we use in the analysis of the algorithm. We prove that the 
graph $F=G\boxplus_{\pi} H$ satisfies these properties with probability $1-o(1)$ in 
Section~\ref{sec:StructPropProof}. 
We first give several definitions. 

\begin{definition}
Consider an SDP solution $\varphi: V_G\to \bbR^n$. We let $\Ball_{\varphi}(u, \delta)$
be the ball of radius $\delta$ around $u$ in the metric induced in $V_G$ by the embedding
$\varphi$:
$$\Ball_{\varphi}(u, \delta) = \{v\in V_G: \|\varphi(u) - \varphi(v)\|^2 \leq \delta\}.$$
For a subset $B\subset V_G$, we let
\begin{equation}\label{def:M}
M_{\xi} (B)= \sum_{v\in V_G} \min\{|N_F(v)\cap B|, \xi\}.
\end{equation}
\end{definition}

In the proof, we need to count the number of vertices in $F$ having at least $\beta d$ neighbors in the $\Ball_{\varphi}(v, 2\delta)$.
Informally, $M_{\beta d}(\Ball_{\varphi}(v,2\delta))$ is an approximation to this number scaled by $\beta d$.
We now state the Main Structural Property.

\begin{property}[Main Structural Property]
There exists a constant $K>0$ (note that $\alpha$, $\beta$ and $D_n$ depend on $K$;
see Section~\ref{sec:prelim}) such that 
for every feasible SDP solution $\varphi: V_G\to \bbR^n$ and $\eta = 2^{-t}$ ($t\leq T$),
there are at most $K\eta dn$ edges
$(u,v)\in E_F$ satisfying the following conditions:
\begin{enumerate}
\item $(u,v)$ is a $\nicefrac{\delta}{2}$-short edge in $\pi(H)$ i.e.,
$\|\varphi(u)-\varphi(v)\|^2 \leq \nicefrac{\delta}{2}$ and $(u,v)\in \pi E_H$.
\item 
$M_{\beta d}(\Ball_{\varphi}(v, 2\delta))\leq \eta \beta d n$.
\item $\short_{\varphi, \nicefrac{\delta}{2}} (u, G)\geq \max\{\beta d, \deg(u, H)/D_n \}$ i.e., there are at least  $\max\{\beta d, \deg(u, H)/D_n \}$ edges
of length $\nicefrac{\delta}{2}$ leaving $u$ in the graph $G$.
\end{enumerate}
\end{property}

In some sense, this is the main property that we need for the proof of~Theorem~\ref{thm:mainSDPbound}. Roughly speaking,
we show that condition 2 is satisfied if $u$ is not a ``heavy vertex'', and
condition 3 is satisfied if $u$ is not a ``damaged vertex''. Hence,  
after removing short edges, heavy vertices, and damaged vertices, 
we obtain a graph ($F_4(t)$) which does not have more than $K\eta_t nd$ edges from $H$. This
implies~Theorem~\ref{thm:mainSDPbound}. Unfortunately, the Damage Control procedure
does not remove all damaged vertices --- it just controls the number of such vertices.
We need Properties 2--4 to show that the edge boundary of the set of the remaining damaged vertices
is small.

The following property is an analog of the Main Structural Property with graphs $G$ and $H$
switched around. Notice that it has an extra condition (4) on edges $(u,v)$ that
are counted.

\begin{property}
For every feasible SDP solution $\varphi: V\to \bbR^n$ and $\eta = 2^{-t}$ ($t\leq T$)
there are at most $K\eta dn$ edges
$(u,v)\in E_F$ satisfying the following conditions:
\begin{enumerate}
\item $(u,v)$ is a $\nicefrac{\delta}{2}$-short edge in $G$, i.e.,
$\|\varphi(u)-\varphi(v)\|^2 \leq \nicefrac{\delta}{2}$ and $(u,v)\in E_G$.
\item $M_{\beta d}(\Ball_{\varphi}(v,  2\delta))\leq \eta \beta d n$.
\item $\short_{\varphi, \nicefrac{\delta}{2}} (u, H)\geq \beta d$, i.e., there are at least $\beta d$ edges
of length $\nicefrac{\delta}{2}$ leaving $u$ in the graph $H$.
\item $\deg(u, G)\leq \alpha d$.
\end{enumerate}
\end{property}

\medskip

\noindent Let 
\begin{equation}\label{eq:VGleqA}
\VGleqA =  \{u\in V_G: \deg(u,G) \leq \alpha d\}
\end{equation}
be the 
set of vertices in $G$ of degree at most $\alpha d$; and let 
\begin{equation}\label{eq:VHgeqB}
\VHgeqB = \{u\in V_G: \deg(u,H) \geq \beta d\}
\end{equation}
be the set of vertices in $H$ of degree at least $\beta d$. 
As we assumed in Section~\ref{sec:prelim}, $|\VGleqA|\leq n/\alpha$ (otherwise, we use an alternative
simple algorithm).
We now state this assumption as Structural Property 3.

\begin{property}
There are at most $n/\alpha$ vertices of degree less than $\alpha d$ in $F$. In other words, $|\VGleqA|\leq n/\alpha$.
Consequently, there are at most $n/\alpha$ vertices of degree less than $\alpha d$ in $G$.
\end{property}

We use this property in several places, particularly to get a bound on the initial total budget: Since we give a budget of $\beta d$ to vertices 
with $\deg (u,F)\geq \alpha d$, and $\alpha d$ to vertices with $\deg (u,F)\leq \alpha d$, the initial budget allocated to vertices is at most 
$\beta d \times n +  (\alpha-\beta)d\times n/\alpha\leq (\beta d + 1)n$. The initial total budget is bounded by
$$(\beta d + 1)n + 3nd/\delta \leq \nicefrac{3}{2}\,\beta d n.$$

\medskip

Finally, we describe the last structural property. This property is rather technical. Roughly speaking, it says that every vertex $u$
has much more neighbors in  $\VHgeqB \setminus \VGleqA$ than in $\VHgeqB \cap \VGleqA$. This happens because $\VHgeqB$ is the image
of the set $\{x\in V_H:\deg(x, H) \geq \beta d\}$ under $\pi$. Every element in $\{x\in V_H:\deg(x, H) \geq \beta d\}$ is much more
likely to be mapped to $V_G\setminus \VGleqA$ than to $\VGleqA$ just because the set $\VGleqA$ is very small.


\begin{property}
For every vertex $u\in V_F$,
$$\sum_{\substack{v:(u,v)\in \pi E_H\\v\in \VHgeqB \cap \VGleqA}} \frac{\beta d}{\deg(v, H)}\leq 
\frac{8}{\alpha}\sum_{\substack{v: (u,v)\in \pi E_H\\v\in \VHgeqB\setminus \VGleqA}} \frac{\beta d}{\deg(v, H)}+4\log n.$$
\end{property}

We prove that the graph $F$ satisfies these Structural Properties w.h.p in  Section~\ref{sec:StructPropProof}. Now we proceed with the analysis of the algorithm.

\subsection{Long Edges Removal}\label{sec:LongEdgesRemoval}
We say that an edge $(u,v)$ is $\delta$-long with respect to the SDP solution $\varphi_t$ if $\|\varphi_t(u)-\varphi_t(v)\|^2 \geq \delta$.
At Step 2 of the main loop of the algorithm, we cut all $\nicefrac{\delta}{2}$-long edges in the graph $F_1(t)$. For every $\nicefrac{\delta}{2}$-long 
edge $(u,v)$ we cut, we increase the budgets of the endpoints of the edge, vertices $u$ and $v$, by 1 (each) and decrease the extra budget (the
variable $\extraBudget$) by 3. This way the total budget decreases by the number of edges cut at this step.
We need to verify that the extra budget is always 
non-negative. To do so, we bound the total number of $\nicefrac{\delta}{2}$-long edges cut during the execution of the algorithm.

\begin{lemma}\label{lem:numberLongEdgesRemoved}
The total number of $\nicefrac{\delta}{2}$-long edges cut by the algorithm is at most $\delta^{-1}dn$.
\end{lemma}
\begin{proof}
At iteration $t$ the algorithm cuts a set $\Upsilon_2(t)$ of $\nicefrac{\delta}{2}$-long edges. Each edge contributes at least $\nicefrac{\delta}{2}$
to $\sdpcost(\varphi_t, F_1(t))$. Once we cut edges in the set $\Upsilon_2(t)$ the SDP value 
decreases by at least $\nicefrac{\delta}{2} |\Upsilon_2(t)|$, i.e. 
$\sdpcost(\varphi_t, F_2(t)) \leq \sdpcost(\varphi_t, F_1(t)) - \nicefrac{\delta}{2} |\Upsilon_2(t)|$.
Observe that $\varphi_t$ restricted to $V_{F_1(t+1)}$ is a feasible (but possibly suboptimal) solution for the graph $F_1(t+1)$. Hence,
\begin{multline*}
\sdpcost(\varphi_{t+1}, F_1(t+1))\leq \sdpcost(\varphi_t, F_1(t+1))\leq\\\leq \sdpcost(\varphi_t, F_2(t))\leq 
\sdpcost(\varphi_{t}, F_1(t)) - \nicefrac{\delta}{2}\cdot|\Upsilon_2(t)|.
\end{multline*}
Thus, $|\Upsilon_2(t)|\leq \nicefrac{2}{\delta}\cdot(\sdpcost(\varphi_{t}, F_1(t))- \sdpcost(\varphi_{t+1}, F_1(t+1)))$, and  
\begin{align*}
\sum_{t=0}^{T-1} |\Upsilon_2(t)| &\leq \nicefrac{2}{\delta}\cdot \sum_{t=0}^{T-1} \sdpcost(\varphi_{t}, F_1(t))-  \sdpcost(\varphi_{t+1}, F_1(t+1))\\
&\leq \nicefrac{2}{\delta}\cdot \sdpcost(\varphi_{t}, F(t))\leq \delta^{-1} dn,
\end{align*}
since the cost of the optimal bisection in graph $F(t)$ is at most $|E_H|\leq dn/2$, and hence 
$\sdpcost(\varphi_{t}, F(t))\leq dn/2$.
\end{proof}

As a corollary we get the following lemma.
\begin{lemma}\label{lem:LongEdgesBudget}
I. Denote by $b_{before}$ the total budget before removing $\nicefrac{\delta}{2}$-long edges; denote by $b_{after}$ the total budget
after removing $\nicefrac{\delta}{2}$-long edges. Then, 
$$b_{after} \leq b_{before} - |\Upsilon_2(t)|.$$

II. The total budget is always non-negative.
\end{lemma}
\begin{proof}
I. Whenever we cut a long edge we increase the budgets of the endpoints by 1 and decrease the extra budget by 3.

II. We never decrease budgets of individual vertices, so their budgets remain positive all the time (note: the total budget of all active vertices may decrease,
because the set of active vertices may decrease). By Lemma~\ref{lem:numberLongEdgesRemoved}, the number of long edges cut is at most $\delta^{-1} dn$, hence
the extra budget may decrease by at most $3\delta^{-1} dn$ (the algorithm uses the extra budget only to pay for cut long edges).
Hence the extra budget is always non-negative.
\end{proof}

\subsection{Heavy Vertices Removal}\label{sec:HeavyVerticesRemoval}
We say that a vertex $u\in V_{F_2(t)}$ is $\eta_t$-heavy if the vertices in the 
ball of radius $3\delta$ around $u$ have budget at least $\beta \eta_t nd$:
$$\budget (\{v: \|\varphi_t(u)-\varphi_t(v)\|^2 \leq 3\delta\})\geq \beta\eta_t nd.$$

The Heavy Vertices Removal procedure sequentially picks vertices $u$ in $V_{F_2(t)}$. If $u$ 
is active (i.e., it was not removed at the current step, or previous steps) and it is an $\eta_t$-heavy
vertex, then we find the radius $r\in [3\delta, 4\delta]$ that minimizes the edge boundary 
$\partial B_u$ of the ball 
$$B_u=\{v \text{ is active}: \|\varphi_t(u)-\varphi_t(v)\|^2 \leq r\}.$$
We remove the set $B_u$ from the graph. Thus, the Heavy Vertices Removal
Step removes a collection of balls $B_u$. The set of cut edges $\Upsilon_3(t)$ is the union
of the corresponding $\partial B_u$.

We need to prove that the procedure satisfies the invariant of the loop: the total budget
decreases by at least $|\Upsilon_3(t)|$. The Heavy Vertices Removal
procedure may remove several components from the graph $F_3(t)$. We verify 
the invariant for each of them independently. 

\begin{lemma}\label{lem:HeavyVerticesBudget}
Consider one of the removed components $B_u$. Let $\partial B_u$ be the edge boundary of the set $B_u$. 

I. Denote by $b_{before}$ the total budget before removing the set $B_u$; and denote by $b_{after}$ the total budget
after removing the set $B_u$. Then, 
$$b_{after} \leq b_{before} - |\partial B_u|.$$

II. The size of the set $B_u$ is at most $3n/4$.
\end{lemma}
\begin{proof}

I. The set $B_u$ contains the ball of radius $3\delta$ around $u$. 
The budget of vertices in this ball is at least $\beta \eta_t nd$,
because $u$ is a heavy vertex. Hence, the budget of $B_u$ is
also at least $\beta \eta_t nd$. After we remove the set $B_u$, the
vertices in $B_u$ are no longer active, so we decrease the 
total budget by at least $\beta \eta_t nd$. 

We now need to bound the size of the
edge boundary $\partial B$. To do so, we use the bound on the cost of the SDP solution.
By Theorem~\ref{thm:mainSDPbound}, 
$$\sdpcost(\varphi_t, F_3(t))\leq \sdpcost(\varphi_t, F_1(t))\leq 8 K\eta_t dn.$$
Since we pick the radius $r$ in the range $[3\delta, 4\delta]$, we get 
by the standard ball growing argument, that the size of the edge boundary 
$|\partial B|$ is at most $8 K\eta_t dn/\delta\leq \nicefrac{\beta}{2}\cdot \eta_t dn$.

After removing the set $B_u$, the total budget decreases by 
$$\budget (B_u)-|\partial B_u|\geq 
\beta \eta_t nd - \nicefrac{\beta}{2}\cdot \eta_t dn
= \nicefrac{\beta}{2}\cdot \eta_t dn\geq |\partial B_u|.$$
Above, we subtract $|\partial B_u|$ from $\budget (B_u)$, because for every cut edge $(v',v'')\in \partial B_u$, $v'\in B_u$, $v''\notin B_u$,
the algorithm increased the budget of $v''$ by 1.

II. We upper bound the size of the $\Ball_{\varphi_t}(u, 4\delta)$ containing the 
set $B_u$.  We apply the SDP spreading constraint~(\ref{SDP:spreading}) for vertex $u$:
$$\sum_{v\in \Ball_{\varphi_t}(u, 4\delta)} (1 - 4\delta) \leq 
\sum_{v\in V_{F_2(t)}} (1 - \|\varphi_t(u) - \varphi_t(v)\|^2) \leq \frac{n}{2}.$$
Using that $\delta = 1/12$ and $(1 - 4\delta) = 2/3$, we get the bound 
$$|\Ball_{\varphi_t}(u, 4\delta)|\leq \nicefrac{3}{4}\; n.$$
\end{proof}

\subsection{Damage Control}\label{sec:DamageControl}
The Damage Control procedure removes components with a small edge boundary and large budget. 
We find a set of vertices $Y\subset V_{F_2(t)}$ to maximize
\begin{equation}\label{eq:flowproblem}
\Delta (Y) \equiv \budget(Y) - 2 | E_{F_3(t)}(Y,\bar Y) |  - 2\beta d |Y|.
\end{equation}
To find the set $Y$ we solve a maximum flow problem on the graph $F_3(t)$ with 
two extra vertices -- the source and the sink. We connect every vertex $u$ in $F_3(t)$
to the source with an edge of capacity $\budget(u)$ and to the sink 
with an edge of capacity $2\beta d$. We set the capacity of every edge in $F_3(t)$ 
to 2. Then we find the minimum cut between the source and the sink. The
set $Y$ is the set of vertices lying in the same part of the cut as the source. It is easy 
to check that $Y$ minimizes~(\ref{eq:flowproblem}). We give the details in 
Appendix~\ref{appendix:damage-control:min-cut}.

If $\Delta(Y) > 0$ we remove the set $Y$ from the graph $F_3(t)$. We denote the edge boundary 
of $Y$ by $\Upsilon_4(t)$; we denote the obtained graph by $F_4(t) = F_3(t) - Y$. Observe
that when we remove the set $Y$ we cut only edges in $\Upsilon_4(t)$. For every 
edge $(u,v)\in \Upsilon_4(t)$, we increase the budget of the endpoint $u$
that we do not remove (i.e.,  $u\notin Y$) by 1.
If $\Delta(Y) < 0$, then we do nothing: We let $F_4(t)=F_3(t)$ and $\Upsilon_4(t)=\varnothing$.

We need to show that for every edge removed from $F_3(t)$ the Damage Control procedure
decreases the total budget by at least 1, and that the size of the set $Y$ is at most $3n/4$. 

\begin{lemma}\label{lem:DamageControlBudget}
I. Let $b_{before}$ be the total budget before applying the Damage Control procedure at step $t$; and 
let $b_{after}$ be the total budget after applying the Damage Control procedure at step $t$. Then, 
$$b_{after} \leq b_{before} - |\Upsilon_4|.$$

II. The size of the set $Y$ is at most $3n/4$.
\end{lemma}
\begin{proof}
If $\Delta(Y)\leq 0$, then the Damage Control procedure does not do anything and thus the statements I and II
are trivial, so we assume $\Delta(Y) \geq 0$.

I. The procedure decreases the total budget by $\budget(Y) - |\Upsilon_4(t)|$: it removes the set $Y$, which
decreases the total budget by $\budget(Y)$; however, for every removed edge $(u,v)\in E_{F_3(t)}(Y,\bar Y)$, $u\in Y$, $v\notin Y$, 
it increases the budget of $v$ by 1, which increase the total budget by  $E_{F_3(t)}(Y,\bar Y)$. Since $\Delta(Y)\geq 0$, we have 
$$\text{``the change in the budget''}=\budget(Y) - |E_{F_3(t)}(Y,\bar Y)|\geq |E_{F_3(t)}(Y,\bar Y)|\equiv|\Upsilon_3(t)|.$$

II. Since $\Delta(Y)\geq 0$, we have $\budget(Y) \geq 2\beta d |Y|$. The budget of the set $Y$ is at most the total 
budget. Initially, the total budget is at most $\nicefrac{3}{2}\,\beta dn$, and during the execution of the algorithm it may 
only decrease (by Lemma~\ref{lem:mainBudget}), so $\budget(Y) \leq \nicefrac{3}{2}\,\beta dn$. Hence, $|Y|\leq \nicefrac{3}{4}\,n$.
\end{proof}

We have established that Step 4 of the algorithm does not violate the invariants of the loop.
We now show that after applying the Damage Control procedure, the boundary of every set $Y'\subset V_{F_4(t)}$
is not too large.

\begin{lemma}\label{lem:DC-budgetBound}
After Step 4, for every $Y'\subset V_{F_4(t)}$,
\begin{equation}\label{eq:Yprime}
\budget(Y') \leq 2 |E_{F_4(t)}(Y',\bar Y')| + 2\beta d |Y'|.
\end{equation}
\end{lemma}
\begin{proof}
Suppose that at Step 4, the algorithm removed a set $Y$ of vertices, and a set $\Upsilon_4(t)$ of edges from $F_3(t)$ (then
$\Upsilon_4(t)$ is the edge boundary of $Y$). Note that the sets $Y$ and $\Upsilon_4(t)$ can possibly be empty. Assume to the contrary that for some set $Y'$ the inequality~(\ref{eq:Yprime}) is violated. We argue that in this case, $\Delta(Y\cup Y')$ would be greater than $\Delta(Y)$ and hence
the Damage Control procedure would remove the set $(Y\cup Y')$ instead of $Y$ from $F_3(t)$. This easily follows from the following observation: The Damage Control
procedure has increased the budget of $Y'$ by the size of the edge boundary between $Y$ and $Y'$ i.e., by $|E_{F_3(t)}(Y,Y')|$. The edge boundary of $Y'$
has decreased by $|E_{F_3(t)}(Y,Y')|$. Hence, before applying the Damage Control procedure, we had
$$
(\budget(Y') + |E_{F_3(t)}(Y,Y')|) > (2 | E_{F_3(t)}(Y',\bar Y') | + 2 |E_{F_3(t)}(Y,Y')|) - 2\beta d |Y'|.
$$
Thus, $\Delta(Y') > 0$, and $\Delta(Y\cup Y')\geq \Delta(Y) + \Delta(Y') > \Delta(Y)$. We get a contradiction with the assumption
that inequality~(\ref{eq:Yprime}) is violated.
\end{proof}

\section{Bounding the Cost of the SDP}\label{sec:costSDP}
In this section, we upper bound the cost of the SDP solution $\sdpcost(\varphi_t, F_1(t))$. We have a trivial upper bound of $OPT\leq dn/2$
for $t = 0$ (as $F_1(0)=F$), so we consider the case $t>0$. In fact, we upper bound $\sdpcost(\varphi, F_4(t))$ for 
the optimal $\varphi$, which equals $\sdpcost(\varphi_{t+1}, F_1(t+1))$. To this end, we show in Lemma~\ref{lem:3part} that 
$F_4(t)$ can be partitioned into 3 balanced pieces with
edge boundary at most $4K \eta_t nd$. As we see in Section~\ref{sec:mainSDPbound}, this immediately 
gives us an upper bound on the cost of the SDP solution: $\sdpcost(\varphi, F_4(t))\leq 4 K\eta_t nd$,
and $\sdpcost(\varphi_t, F_1(t+1)) \leq 8 K\eta_{t+1} nd$.

\subsection{Partitioning into Three Balanced Sets} \label{sec:three-sets}
Define a new feasible SDP solution $\varphi'_t$ for $F$ i.e., a mapping of $V\to \bbR^n$
satisfying the SDP constraints. Let
\begin{equation}\label{eq:phiPrime}
\varphi'_t(u)=
\begin{cases}
\varphi_t(u), \text{if } u\in V_{F_3(t)};\\
e_u, \text{otherwise};
\end{cases}
\end{equation}
where $e_u$ is a vector of length $\nicefrac{\sqrt{2}}{2}$ orthogonal to all other vectors in the SDP solution (including other $e_v$'s).
This SDP solution coincides with $\varphi_t$ on the set of active vertices. Note that
all edges in $F_4(t)$ are $\nicefrac{\delta}{2}$-short w.r.t. $\varphi'_t$, and any edge connecting an active
vertex and a removed vertex has length 1.

We now show that all active vertices $u\in V_{F_4(t)}$ satisfy the second condition on edges $(u,v)$ in Property~1 and in Property~2
for the SDP solution $\varphi'_t$.

\begin{lemma}\label{lem:prop2isSatisfied}
For all vertices $u\in V_{F_4(t)}$, 
$M_{\beta d}(\Ball_{\varphi'_t}(u, 2\delta))\leq \eta_t \beta d n$, where $\varphi'_t$ is defined in (\ref{eq:phiPrime}).
\end{lemma}
\begin{proof}
Fix a vertex $u\in V_{F_4(t)}$. Let 
$B_{2\delta} = \Ball_{\varphi'_t}(u, 2\delta)$ and 
$$B_{3\delta} = \Ball_{\varphi_t}(u, 3\delta)\cap V_{F_3(t)} = \{v\in V_{F_3(t)}: \|\varphi_t(v) - \varphi_t(u)\|^2\leq 3\delta\}.$$
That is,
$B_{2\delta}$ is the ball of radius $2\delta$ around $u$ in the graph $F_4(t)$;
$B_{3\delta}$ is the ball of radius $3\delta$ in the graph $F_3(t)$ w.r.t $\varphi_t$ (and not
$\varphi'_t$). Note that $B_{2\delta}\subset B_{3\delta}$. Write the definition of 
$M_{\beta d}(\Ball_{\varphi'_t}(u, 2\delta))= M_{\beta d}(B_{2\delta})$:
\begin{align}
M_{\beta d}(B_{2\delta}) &= \sum_{v\in V_G} \min\{|N_F(v)\cap B_{2\delta}|, \beta d\}\notag\\
&= \sum_{v\in B_{3\delta}} \min\{|N_F(v)\cap B_{2\delta}|, \beta d\}\\
{}&+\sum_{v\in V_G\setminus B_{3\delta}} \min\{|N_F(v)\cap B_{2\delta}|, \beta d\}.\label{eq:Mbd}
\end{align}
Observe that all vertices in $B_{2\delta}$ are active, since the distance to already removed
vertices equals 1 (see~(\ref{eq:phiPrime})). We separately bound the first and second sums above.
We bound the size of the first sum by $|B_{3\delta}|\beta d$. 
To bound the size of the second sum, consider a vertex $v\in V_G\setminus B_{3\delta}$.
There are two options: 
\begin{enumerate}
\item  $v\notin V_{F_3(t)}$ i.e. $v$ was removed at one of the previous 
iterations or at Step 3. In this case, all edges going from $v$ to $B_{2\delta}$ 
were cut at one of the previous iterations or at Step 3.
\item  $\|\varphi_t(u)-\varphi_t(v)\|^2\geq 3\delta$ 
(but $v\in V_{F_3(t)}$). In this case,
all edges going from $v$ to $B_{2\delta}$ have length at least $\delta$, 
and thus they were cut at Step 2 at one of the iterations.
\end{enumerate}
In any case, all edges between $v$ and vertices in $B_{2\delta}$ 
have been cut by the algorithm at Steps 2, 3 of the current iteration, or at any step of
one of the previous iterations. Note that none of these edges were cut at Step 4 of the current
iteration.  Let $\rho$ be the number of edges going from $V_G\setminus B_{3\delta}$ to
$B_{2\delta}$. Then, $\rho$ is an upper bound on the second sum in (\ref{eq:Mbd}).
We have
$$M_{\beta d}(B_{2\delta})\leq |B_{3\delta}|\beta d + \rho.$$
We know that every cut edge has increased the budget of the endpoint lying in $B_{2\delta}$ by 1. Initially, 
the algorithm assigned a budget of at least $\beta d$ to each vertex in $B_{3\delta}$, hence
$$\budget(B_{3\delta})\geq \beta d\cdot |B_{3\delta}| + \rho \geq M_{\beta d}(B_{2\delta}).$$
In the equation above, we compute the budget of $B_{3\delta}$ after Step 3. That is,
we ignore the changes of the budgets that occurred at Step~4.

We now use that the Heavy Vertices Removal procedure has removed all balls of radius 
$3\delta$ having a budget of $\eta_t \beta  dn$ or more. Thus, $\budget(B)\leq \beta \eta_t dn$ 
(again, here we compute the budget after Step 3). We conclude that 
$M_{\beta d}(B_{2\delta})\leq \eta_t \beta  dn$.
\end{proof}
We now state the main technical result of this section. 
\begin{lemma}\label{lem:3part} 
Let $(L,R)$ be the planted partition in the graph $G$. For every $t\in \{0,\dots, T-1\}$, the graph $F_4(t)$ can be partitioned into
two sets $W$, $\bar{W}$ such that the sizes of the 
sets $L\cap W$, $R\cap W$ and $\bar{W}$ are at most $n/2$ each; and the size of the edge boundary between $L\cap W$, $R\cap W$ and $\bar{W}$
is at most $4K \eta_t dn$.
\end{lemma}
\begin{proof}
Define the set $W$ as follows:
$$W=\{u\in V_{F_4(t)}: \deg (u, G_4(t)) \geq \max\{\beta d, \deg(u, H)/D_n\} \}.$$

Let $\varphi'_t$ be the SDP solution defined in~(\ref{eq:phiPrime}). We claim that all edges $(u,v)$ in $E_{H_4(t)}$ 
with $u\in W$ satisfy conditions 1--3 of the Main Structural Property. Indeed, all edges $(u,v)\in E_{F_4(t)}$ are 
$\nicefrac{\delta}{2}$-short, otherwise they would be removed by Step 2. By Lemma~\ref{lem:prop2isSatisfied}, all active vertices 
satisfy the second condition. Finally, by the definition of $W$, the degree of every $u\in W$ in the graph $G_4(t)$ is 
at least $\max\{\beta d,\deg(u, H)/D_n\}$, and since all uncut edges are $\nicefrac{\delta}{2}$-short, $\short_{\varphi'_t, \nicefrac{\delta}{2}} (u, G)\geq \max\{\beta d,\deg(u, H)/D_n\}$. Therefore, by the Main Structural Property, there are at most $K\eta_t dn$ edges $(u,v)\in E_{H_4(t)}$ with $u\in W$. 

The edge boundary between the sets $L\cap W$, $R\cap W$ and $\bar{W}=V_{F_4(t)}\setminus W$ is the union of the
sets $E_{G_4(t)}(L\cap W, R\cap W)$, $E_{G_4(t)}(W,\bar W)$, $E_{H_4(t)}(L\cap W, R\cap W)$ and $E_{H_4(t)}(W, \bar W)$.
Observe, that $E_{G_4(t)}(L\cap W, R\cap W)=\varnothing$, since $(L,R)$ is the planted cut in $G$.
We already have an upper bound 
\begin{equation}\label{eq:bound-W-boundary-in-H}
|E_{H_4(t)}(L\cap W, R\cap W)| + |E_{H_4(t)}(W, \bar W)|\leq K\eta_t dn
\end{equation} 
since all edges in $E_{H_4(t)}(L\cap W, R\cap W)$ and  $E_{H_4(t)}(W, \bar W)$ are incident on the set $W$. We now bound the size
of the set $E_{G_4(t)}(W,\bar W)$. Let
\begin{align*}
X  &= \{u\in V_{F_4(t)}: \beta d\leq \deg (u, G_4(t)) \leq \deg(u,H)/D_n\};\\
Y  &= \{u\in V_{F_4(t)}: \deg (u, G_4(t)) \leq \beta d; \budget(u) \geq (\alpha - \beta) d\};\\
Z  &= \{u\in V_{F_4(t)}: \deg (u, G_4(t)) \leq \beta d; \budget(u) < (\alpha - \beta)d\}.
\end{align*}
The budgets in the expressions above are computed in the end of the $t$-th iteration. We will need several bounds on the 
degrees of vertices $u$ in $Z$.

\begin{claim}\label{cl:degsOfYZ}
For every $u\in Z$,
\begin{enumerate}
\item $\deg (u,H_4(t))\geq \max\{\beta d, \nicefrac{1}{6}\deg (u,H)\}$;
\item $\deg (u,G)\leq (\alpha-\beta) d$.
\end{enumerate}
\end{claim}
\begin{proof}
I. Denote by $\rho=\deg (u,F) -\deg(u,F_4(t))$ the number of cut edges incident 
on $u$. By the definition of $Z$, $\budget (u)< (\alpha - \beta)d$. Therefore, 
$\deg(u, F)\geq \alpha d$ (otherwise, $u$ would receive a budget of $\alpha d$ at the initialization step);
and the initial budget of $u$ is $\beta d$. Hence, the current budget of $u$ equals
$\budget(u)=\beta d + \rho$. We get $\rho =\budget(u) -\beta d \leq (\alpha - 2\beta) d$. Thus,
$\deg(u,F_4(t)) = \deg (u,F)-\rho \geq 2\beta d$. Then,
$\deg(u,H_4(t))\geq \deg(u,F_4(t))- \deg(u,G_4(t))\geq \beta d$.

By Lemma~\ref{lem:DC-budgetBound} (applied with $Y'=\{u\}$), 
$$\budget (u)\leq 2\deg (u, F_4(t)) + 2\beta d\leq 2 \deg (u, H_4(t)) +4\beta d.$$
Hence, $\rho \leq 2 \deg (u, H_4(t)) +3\beta d$, and
$$\deg(u,H)\leq \deg(u,H_4(t)) + \rho \leq 3 \deg (u, H_4(t)) +3\beta d\leq 6 \deg (u, H_4(t)).$$

II. We have $\deg (u,G)\leq \deg(u, G_4(t))+\rho \leq \beta d + (\alpha - 2\beta) d\leq (\alpha -\beta)d$.
\end{proof}
As an immediate corollary we get that $Z\subset \VGleqA\cap \VHgeqB$ (see (\ref{eq:VGleqA}) and (\ref{eq:VHgeqB}) in Section~\ref{sec:StructPropDefs} for the definitions
of $\VGleqA$ and $\VHgeqB$).
\begin{corollary}\label{cor:degsOfYZ} $Z\subset \VGleqA\cap \VHgeqB$.
\end{corollary}

We will also need an upper bound on the size of $X$.
\begin{claim}\label{cl:sX}
$|X|\leq n/(\beta D_n)$.
\end{claim}
\begin{proof}
For all $u\in X$, $\deg(u,H)\geq \beta d$. The average degree of vertices in $H$ is at most $d$. Hence, by Markov's inequality, $|X|\leq n/(\beta d)$.
\end{proof}

\smallskip

We return to the proof of Lemma~\ref{lem:3part}. Write
$$|E_{G_4(t)}(W,\bar W)|= |E_{G_4(t)}(Y, W)| + |E_{G_4(t)}(Z,W)|.$$

Observe that every edge $(u,v)$ in $E_{G_4(t)}(Z, W)$ ($u\in Z$, $v\in W$) satisfies conditions 1--4 of 
Structural Property~2: Each edge $(u,v)$ is $\nicefrac{\delta}{2}$-short. Then,
$M_{\beta d}(\Ball_{\varphi}(v, 2\delta))\leq \eta_t n$ by Lemma~\ref{lem:prop2isSatisfied} (note that $v$ is active, because $(u,v)$ is $\nicefrac{\delta}{2}$-short);
$\short_{\varphi'_t,\nicefrac{\delta}{2}}(u,G)\geq \deg(u,G_4(t))\geq \beta d$ by Claim~\ref{cl:degsOfYZ};
and $\deg(u,G)\leq \alpha d$ by Claim~\ref{cl:degsOfYZ}. Hence,
\begin{equation}\label{eq:bound-on-ZW}
|E_{G_4(t)}(Z,W)|\leq K\eta_t nd.
\end{equation}

Claim~\ref{claim:zetaSize} shows that $|E_{G_4(t)}(Y, W)|\leq K\eta_t nd$, and  
Claim~\ref{cl:Xboundary} shows that $|E_{G_4(t)}(X,\bar W)|\ll K\eta_t nd$. 
Hence, the total size of the 
edge boundary is at most $4K\eta_t dn$. Before proving Claims~\ref{claim:zetaSize} and~\ref{cl:Xboundary}, we verify that the sizes of
the sets $L\cap W$, $R\cap W$ and $\bar W$ are bounded by $n/2$. The sizes of the sets 
$L\cap W$, $R\cap W$ are bounded by $|L|=|R|=n/2$. The size of the set $X$ is bounded by $n/(\beta d)\ll n/6$ (see Claim~\ref{cl:sX}); the size of $Y$ is bounded by $\totalBudget/((\alpha-\beta)d)\leq
\nicefrac{3}{2}\,\beta dn /((\alpha-\beta)d)\ll n/6$ (since the budget of every vertex in $Y$ is at least $(\alpha-\beta)d$);
the size of $Z$ is bounded by $n/\beta\ll n/6$ (since the average degree in $H$ is at most $d$; 
the degrees of all vertices in $Y$ are at least $\beta d$). Thus, $|\bar W| = |X|+|Y| + |Z|\ll n/2$.

\end{proof}

\begin{claim}\label{cl:Xboundary}
The size of the edge boundary between $W$ and $X$ in the graph $G_4(t)$ is at most $\eta_t nd$:
$$|E_{G_4(t)}(W,X)|\leq \eta_t nd.$$
\end{claim}
\begin{proof}
We count the number of edges incident on the vertices of $X$ in the graph $G_4(t)$. By the definition of $X$,
$\deg(u, G_4(t))\leq \deg(u,H)/D_n$. Thus,
$$\sum_{u\in X}\deg(u, G_4(t))\leq \sum_{u\in X}\frac{\deg(u,H)}{D_n}\leq \frac{2|E_H|}{D_n} \leq \frac{dn}{D_n}\leq \eta_t nd,$$
since $\eta_t\geq 1/D_n$.
\end{proof}

To prove Claim~\ref{claim:zetaSize} we need to bound $\beta d |Z|$.

\begin{claim}\label{claim:sizeBetaDZ} We have
$$ 
\beta d |Z| \leq \frac{\budget (\bar W)}{10} + 7 |E_{H_4(t)}(\bar W, W)|.
$$
\end{claim}
\begin{proof}
Consider the following mental experiment: \textit{We give $\beta d$ blue tokens to every $z\in Z$, and $\beta d$ red tokens
to every $y\in X\cup Y$. Then, every vertex $x\in \bar W$ sends $\beta d/\deg(x,H)$ tokens to each neighbor of $x$ in the graph $H$.}
Let us write that the total number of blue tokens, $\beta d |Z|$, equals the number of tokens sent from vertices in $Z$:
$$\beta d |Z|= \sum_{y\in Z} \deg(y, H) \times \frac{\beta d}{\deg(y, H)}.$$
Using the bound $\deg (z,H_4(t))\geq \nicefrac{1}{6}\deg (z,H)$ from Claim~\ref{cl:degsOfYZ}, we get
$$
\beta d |Z|\leq 6 \sum_{y\in Z} \deg (z,H_4(t)) \times \frac{\beta d}{\deg(y, H)} = 
6\sum_{y\in Z}\;\sum_{x:(x,y)\in E_{H_4(t)}} \frac{\beta d}{\deg(y, H)}.
$$
Now, in the right hand side, we count the number of blue tokens sent along edges in $H_4(t)$. Observe that every vertex $y$ sends
at most one blue token to each of its neighbors --- simply because $\deg(y, H)\geq \beta d$ (see Claim~\ref{cl:degsOfYZ}). Hence, 
the number of tokens sent to $W$ from all $y$'s in $Z$ is bounded by the size of the edge boundary between $Z$ and $W$. We have
\begin{eqnarray}\label{eq:blue-tokens}
\beta d |Z|&\leq& 
6\sum_{y\in Z}\;\sum_{\substack{x:(x,y)\in E_{H_4}(t)\\x\in \bar W}} \frac{\beta d}{\deg(y, H)} + 6 |E_{H_4(t)}(Z, W)|
\\&=&
6\sum_{x\in \bar W}\Big(\sum_{\substack{y:(x,y)\in E_{H_4}(t)\\y\in Z}} \frac{\beta d}{\deg(y, H)} \Big) + 6 |E_{H_4(t)}(Z, W)|.\notag
\end{eqnarray}
The expression in the brackets above is the number of blue tokens a vertex $x\in \bar W$ receives.
We compare it with the number of red tokens received by the same vertex. Using Structural Property 4,
we get (keep in mind that $Z\subset \VGleqA\cap \VHgeqB$, see Corollary~\ref{cor:degsOfYZ})
\begin{multline*}
\sum_{\substack{y:(x,y)\in \pi E_H\\y\in Z}} \frac{\beta d}{\deg(y, H)}\leq 
\sum_{\substack{y:(x,y)\in \pi E_H\\y\in \VGleqA\cap \VHgeqB}} \frac{\beta d}{\deg(y, H)}\leq \\ \leq
\frac{8}{\alpha}\sum_{\substack{y: (x,y)\in \pi E_H\\y\in \VHgeqB\setminus \VGleqA}} \frac{\beta d}{\deg(y,H)}+4\log n \leq 
\frac{8}{\alpha}\sum_{\substack{y: (x,y)\in \pi E_H\\y\in \VHgeqB\setminus Z}} \frac{\beta d}{\deg(y, H)}+4\log n.
\end{multline*}
We cover the domain $\{y: (x,y)\in \pi E_H \text{ and } y\in \VHgeqB\setminus Z\}$ with three
sets $S_1=\{y: (x,y)\in \pi E_H\setminus E_{H_4(t)} \text{ and } y\in \VHgeqB\}$, 
$S_2=\{y: (x,y)\in E_{H_4(t)} \text{ and } y\in W\}$ and $S_3=\{y: (x,y)\in E_{H_4(t)} \text{ and } y\in X\cup Y\}$.
The size of $S_1$ is at most $\budget(x)-\beta d$, since all edges from $S_1$ have been cut
and hence $\budget(x)\geq \beta d + |S_1|$. The set $S_2$ equals $E_{H_4(t)} (\{x\}, W)$. Therefore, 
using that $\beta d/\deg(y,H)\leq 1$, we get
\begin{align*}
\sum_{\substack{y:(x,y)\in \pi E_H\\y\in Z}} \frac{\beta d}{\deg(y, H)}
&\leq
\frac{8}{\alpha}\Big(|S_1|+|S_2| + \sum_{y\in S_3} \frac{\beta d}{\deg(y, H)} \Big)+4\log n\\
&\leq\frac{8}{\alpha}\Big(\sum_{\substack{y: (x,y)\in E_{H_4(t)}\\y\in X\cup Y}} \frac{\beta d}{\deg(y, H)}+ |E_{H_4(t)} (\{x\}, W)| + \budget (x)\Big) + 4\log n.
\end{align*}
Plugging this inequality in~(\ref{eq:blue-tokens}), we get
\begin{align*}
\beta d |Z|&\leq
\frac{48}{\alpha}\sum_{x\in \bar W}
\Big(\sum_{\substack{y: (x,y)\in E_{H_4(t)}\\y\in X\cup Y}} \frac{\beta d}{\deg(y, H)}+ \budget (x) + E_{H_4(t)} (\{x\}, W) \Big)\\
& \phantom{\leq\frac{48}{\alpha}\sum}
\;+\;\Big(4|\bar W| \log n	
 + 6 |E_{H_4(t)}(Z, W)|\Big)\\
&\leq \frac{48}{\alpha}\Big(\beta d |X| + \beta d |Y| + \budget (\bar W) + |E_{H_4(t)} (\bar W, W)|\Big) + 4|\bar W| \log n
 + 6 |E_{H_4(t)}(Z, W)|.
\end{align*}
Using that $\alpha = 50\beta$,  $\beta = 200K$, $d>\beta \log n$, and $\budget (X)\geq \beta d|X|$,
$\budget (Y)\geq (\alpha - \beta )d|Y|$, $\budget (Z)\geq \beta d|Z|$,
we get the following bounds 
\begin{itemize}
\item $48\beta d |X|/\alpha + 48\budget (X)/\alpha + 4|X|\log n\leq \budget(X)/10$;
\item $48\beta d |Y|/\alpha + 48\budget(Y)/\alpha+4|Y|\log n\leq \budget(Y)/10$; and
\item $48\budget(Z)/\alpha +4|Z| \log n\leq \budget(Z)/10$.
\end{itemize}
 Therefore,
$$\beta d |Z| \leq \frac{\budget (\bar {W})}{10} + 7 |E_{H_4(t)}(\bar W, W)|.$$
\end{proof}

We now use the upper bound on $\beta d |Z|$ to get an upper bound on $\zeta\equiv |E_{G_4(t)}(W,Y)|$.

\begin{claim}\label{claim:zetaSize}
The size of the edge boundary between $W$ and $Y$ in the graph $G_4(t)$ is at most $K\eta_t nd$:
$$\zeta\equiv |E_{G_4(t)}(W,Y)|\leq K\eta_t nd.$$
\end{claim}
\begin{proof}
The Damage Control procedure ensures that for $Y' = \bar W$  (see Lemma~\ref{lem:DC-budgetBound}), we have 
$$
\budget (\bar W) \leq 2|E_{F_4(t)}(\bar W, W)| + 2\beta d |\bar W| = 
2|E_{F_4(t)}(\bar W, W)| + 2\beta d |X| + 2\beta d |Y| + 2\beta d| Z|.
$$
We bound the term $\beta d |X|$ using Claim~\ref{cl:sX} and the term $\beta d |Z|$ using Claim~\ref{claim:sizeBetaDZ}. 
We get the following bound.
$$
\budget (\bar W) \leq 2|E_{F_4(t)}(\bar W, W)| + 2\beta d |Y| + \frac{\budget (\bar W)}{5} + 
14 |E_{H_4(t)}(\bar W, W)| + \frac{2dn}{D_n}.
$$
Hence,
$$
\budget (\bar W) \leq \nicefrac{5}{2}\beta d |Y| + \nicefrac{5}{2} |E_{F_4(t)}(\bar W, W)| + 18 |E_{H_4(t)}(\bar W, W)| + \frac{3dn}{D_n}.
$$
We replace $|E_{F_4(t)}(\bar W, W)|$ with $|E_{G_4(t)}(\bar W, W)|+|E_{H_4(t)}(\bar W, W)|$,
$$\budget (\bar W) \leq \nicefrac{5}{2}\,\beta d |Y| + \nicefrac{5}{2}\,|E_{G_4(t)}(\bar W, W)| 
+ 21 |E_{H_4(t)}(\bar W, W)|+ \frac{3dn}{D_n}.$$
Recall, that every vertex in $Y$ has a budget of at least $(\alpha - \beta) d$ (by the definition of $Y$). Thus,
$\budget(\bar{W})\geq \budget(Y)\geq (\alpha - \beta) d |Y|$. We get
$$(\alpha - \nicefrac{7}{2}\,\beta) d |Y|\leq   \nicefrac{5}{2}\,|E_{G_4(t)}(\bar W, W)| + 
21 |E_{H_4(t)}(\bar W, W)| + \frac{3dn}{D_n}.$$
The degree of every vertex $u$ in $Y$ in the graph $G_4(t)$ is 
at most $\beta d$ (by the definition of $Y$). Hence, $\zeta \equiv |E_{G_4(t)}(Y,W)|\leq \beta d |Y|$, and
$$
\zeta \leq \frac{\beta}{\alpha-\nicefrac{7}{2}\,\beta}\times \Big[\nicefrac{5}{2}\,|E_{G_4(t)}(\bar W, W)| +
 21 |E_{H(t)}(\bar W, W)| 
+\frac{3dn}{D_n}\Big].
$$
Finally, we use the inequalities $|E_{H_4(t)}(\bar W, W)|\leq K\eta_t dn$ (see (\ref{eq:bound-W-boundary-in-H})); $3dn/D_n\leq 3\eta_t dn$; and
$$|E_{G_4(t)}(\bar W, W)| = |E_{G_4(t)}(X, W)| + |E_{G_4(t)}(Y, W)| + |E_{G_4(t)}(Z, W)|\leq \eta_t dn + \zeta + K\eta_t dn$$
(see Claim~\ref{cl:Xboundary} and Equation~(\ref{eq:bound-on-ZW})) to obtain the bound:
$$
\zeta \leq \frac{\beta}{\alpha-\nicefrac{7}{2}\,\beta}\times \Big[\nicefrac{5}{2}\zeta + 24 K\eta_t dn\Big],
$$
which implies $\zeta < K \eta_t dn$.
\end{proof}

\subsection{Proof of~Theorem~\ref{thm:mainSDPbound}}\label{sec:mainSDPbound}
\begin{proof}[Proof of Theorem~\ref{thm:mainSDPbound}]
Let $L\cap W$, $R\cap W$ and $\bar W$ be the partitioning from Lemma~\ref{lem:3part}. We pick three orthogonal
vectors $e_L$, $e_R$ and $e_{\bar{W}}$ of lengths $\nicefrac{\sqrt{2}}{2}$. We define a new SDP solution $\varphi: V_{F_4(t)}\to \bbR^n$ as follows.
Let $\varphi(u) = e_L$ for $u\in L\cap W$; $\varphi(u) = e_R$ for $u\in R\cap W$ and $\varphi(u)=e_{\bar{W}}$ for $u\in \bar W$. It is easy to
check that this SDP solution is feasible: it trivially satisfies the $\ell_2^2$-triangle inequalities (since it is a 0-1 metric); and it satisfies 
the spreading constraints since the sets $L\cap W$, $R\cap W$ and $\bar W$ are balanced. The cost of the solution, $\sdpcost(\varphi, F_4(t))$ 
exactly equals the number of edges cut by the partition, which is bounded by $4K\eta_t dn = 8K\eta_{t+1} dn$:
$$
\sdpcost(\varphi_{t+1}, F_1(t+1))\leq \sdpcost(\varphi, F_1(t+1)) 
= \sdpcost(\varphi, F_4(t))\leq 8K\eta_{t+1} dn.
$$
\end{proof}

\section{Structural Properties --- Proofs}\label{sec:StructPropProof}

In this section we show that $F=G\boxplus_{\pi} H$ satisfies Structural Properties 1--4 with high probability (see
Section~\ref{sec:StructPropDefs} for definitions). The main technically interesting and conceptually important part of our proof is
the Main Structural Theorem.

\begin{theorem}[Main Structural Theorem]\label{thm:MainStructuralTheorem}
There exist a constant $K$, such that for every $\beta>1$, $D_n>\sqrt{\log n}$, and $d \geq D_{ARV} D_n\log^2_2 n$,
every graphs $G=(V_G,E_G)$ and $H=(V_H,E_H)$ on vertex sets $V_G=L_G\cup R_G$ and $V_H=L_H\cup R_H$ with $|E_H|\leq dn/2$ and 
$|L_G|=|R_G|=|L_H|=|R_H|=n/2$, the following statement holds with probability $(1- 1/n^2)$
for $\pi$ chosen uniformly at random from $\Pi_{LR}$. 
For every feasible SDP solution $\varphi: V_G\to \bbR^n$ and $\eta = 2^{-t}$ ($t\leq T=\roundup{\bLog D_{ARV}}=O(\log \log n)$),
there are at most $K\eta dn$ elements in the set $\calS$ defined as follows: the elements of $\calS$ are ordered pairs; a pair 
$(u,v) \in V_G \times V_G$ belongs to $\calS$, if
\begin{enumerate}
\item $(u,v)$ is a $\nicefrac{\delta}{2}$-short edge in $\pi(H)$ i.e.,
$\|\varphi(u)-\varphi(v)\|^2 \leq \nicefrac{\delta}{2}$ and $(u,v)\in \pi E_H$.
\item 
$M_{\beta d}(\Ball_{\varphi}(v, 2\delta))\leq \eta \beta d n$.
\item $\short_{\varphi, \nicefrac{\delta}{2}} (u, G)\geq \max\{\beta d, \deg(u, H)/D_n \}$ i.e., there are at least  $\max\{\beta d, \deg(u, H)/D_n \}$ edges
of length $\nicefrac{\delta}{2}$ leaving $u$ in the graph $G$.
\end{enumerate}
\end{theorem}

Conditions 1--3 in the statement of the theorem are the same as in the Main Structural Property. Note that we set $\beta$ to
$200K$  in the algorithm, and thus $\beta$ depends on $K$. In this theorem, we prove that a universal constant $K$ exists
that works for every $\beta > 1$ and $D_n > \sqrt{\log n}$, particularly,
for $\beta = 200K$, $\alpha = 50\beta$, and $D_n = \max\{D_{ARV}, \alpha\}$. We do not assume
that $G$ has a planted cut i.e., some edges in $G$ may cross the cut $(L_G, R_G)$.

In the proof, we use the notion of prefix Kolmogorov complexity. We denote the complexity of $x$ given $G$ and $H$ by%
\footnote{In the standard notation our $KP(x)$ is $KP(x\given G,H)$.}
 $KP(x)$. The reader may interpret 
the statement $KP(x)=\ell$ as follows: we can encode $x$ using $\ell$ bits, such that this encoding can be uniquely decoded to $x$ given $G$ and $H$. 
Particularly, 
$KP(\pi)$ is the complexity of the bijection $\pi \in \Pi_{LR}$ i.e., $KP(\pi)$ is the number of bits required to store $\pi$.

\subsection{Proof of the Main Structural Theorem}

\begin{proof}
Fix $\eta=2^{-t}$ and an SDP solution $\varphi$. Let $\gamma = |\calS|/(dn)$. We show that if $\gamma > K \eta$ (for some  constant $K$), then 
the permutation $\pi$ can be encoded with a binary string of length less than $\bLog |\Pi_{LR}|-2\bLog n$. In other words, the prefix Kolmogorov complexity
of $\pi$ is at most $\bLog |\Pi| - 2 \bLog n$. This is an unlikely event for a random $\pi$ sampled 
from $\Pi_{LR}$ uniformly. So we conclude that $|\calS|\geq K\eta dn$ with small probability.

To construct the encoding we need to identify a set of vertices $x\in \calX$ for which the description of $\pi(x)$ is short. 
Denote the set of vertices whose degrees are in the range $[2^i, 2^{i+1})$ in $H$ by $\calU_i$:
$$\calU_i=\{x\in V_H: \deg (x, H) \in [2^i,2^{i+1}-1]\}.$$
Let $\calE_i=\{(x,y)\in E_H: x\in \calU_i\}$ and let $\calS_i = \{(x,y)\in \calS: x\in \calU_i\}$. The pairs $(x,y)$ in $\calE_i$ and in $\calS_i$
are ordered pairs. Let $\lambda_i = |\calE_i|/(dn)$, and $\gamma_i = |\calS_i|/|\calE_i|$. Then,
\begin{equation}\label{eq:sumLambda}
\gamma = \sum_i \lambda_i \gamma_i.
\end{equation}
Note that $|E_H|\leq dn/2$ (by the definition of $d$), thus
$\sum \lambda_i \leq 1$. Consider the set of indices 
$I=\{i:\lambda_i\gamma_i\geq \gamma/(2\bLog n)\}$.
We have 
$$\sum_{i\in I} \lambda_i\gamma_i\geq \gamma - \sum_{i:\lambda_i \gamma_i\leq \gamma/(2\bLog n)}\lambda_i \gamma_i\geq \gamma/2.$$
We pick one $i\in I$ with $\gamma_i\geq \gamma/2$.

To encode $\pi$, we need to store the embedding $\varphi$. However, we cannot afford to store the whole embedding, so we only encode the 
embeddings of two subsets $Q_G\subset V_G$ and $Q_H\subset V_H$ of size at most $3qn$, where $q=D_n/2^i$. For every $x\in \calU_i$, let $x'$ be the first
element in $Q_H$ according to the order in $Q_H$ that is a neighbor
of $x$ in the graph $H$ i.e. $(x,x')\in E_H$. If $x'$ exists, then we define two sets $\Xi'(x)$ and $\Xi''(x)$ as follows: 
\begin{align*}
\Xi'(x)&=\{u: |Q_G\cap N_G(u)\cap \{v: \|\varphi(v)-\varphi(\pi(x'))\|^2 \leq \delta\}| \geq q\beta d / 2\};\\
\Xi''(x)&=\{u: |Q_G\cap N_G(u)\cap \{v: \|\varphi(v)-\varphi(\pi(x'))\|^2 \leq 2\delta\}| \geq q\beta d/2\}.
\end{align*}
The only difference in the definitions of $\Xi'(x)$ and $\Xi''(x)$ is that the radius of the ball around $\varphi(\pi(x'))$ is $\delta$ for $\Xi'(x)$ 
and $2\delta$ for $\Xi''(x)$. Note that $\Xi'(x)\subset \Xi''(x)$. Consider the set $\calX$ of vertices $x$ for which the following conditions hold:
\begin{enumerate}
\item $x'$ is defined and $(x,x')$ is a $\nicefrac{\delta}{2}$-short edge w.r.t. $\varphi$;
\item $\pi(x)\in \Xi'(x)$;
\item $|\Xi''(x)|\leq 4 \eta n$.
\end{enumerate}

We show that for the right choice of $Q_G$ and $Q_H$, the set $\calX$ is sufficiently large. 
\begin{lemma}\label{lem:existQ}
There exist sets $Q_G\subset V_G$ and $Q_H\subset V_H$ such that $|\calX|\geq c \gamma_i |\calU_i|$ (for some absolute constant~$c$).
\end{lemma}
\begin{proof}
Let $Q_G$ and $Q_H$ be random subsets of $V_G$ and $V_H$ such that every $u\in V_G$ belongs to $Q_G$,
and every $x\in V_H$ belongs to $Q_H$ with probability $q$. Then, by Chernoff's bound, the probability that $|Q_G|,|Q_H|\leq 3qn$ is at least 
$$1-2e^{-qn}\geq 1 - 2e^{-D_n}\geq 1 - 2e^{-\sqrt{\log n}}.$$ 
We now estimate the expected size of $\calX$. 

The degrees of all vertices in $\calE_i$ are at least $2^{i}$. Hence, the probability that a vertex $x\in \calU_i$ does not have neighbors in $Q_H$ is 
at most $(1-D_n/2^i)^{2^i}\ll \nicefrac{1}{4}$.
Consequently, for every $x\in \calU_i$, $x'$ is defined with probability at least $\nicefrac{3}{4}$. The 
edge $(x,x')$ belongs to $\calS_i$ with probability $|N_H(x)\cap \calS_i|/|N_H(x)|\geq |N_H(x)\cap \calS_i|/2^{i+1}$.
Note that the event $(x,x')\in \calS_i$ depends only on the set $Q_H$. Let us condition on $Q_H$ and assume that $(x,x')\in \calS_i$.

Let $B(x)$ be the set of neighbors of $\pi(x)$ in $G$ connected to $\pi(x)$ via $\nicefrac{\delta}{2}$-short edges. 
Since $(x,x')\in S$, the set $B(x)$ has at least $\max\{\beta d, 2^i/D_n\}$ vertices (by condition 3 of the definition of $\calS$).
The size of $B(x)\cap Q_H$ is distributed as a Binomial distribution with parameters $|B(x)|$ and $q$. The median of
the distribution is at least $\rounddown{q|B(x)|}$. Note that $\rounddown{q|B(x)|}\geq 1$, since $|B(x)| \geq 2^i/D_n$.
Hence, $\rounddown{q|B(x)|}\geq q\beta d /2$. Therefore,
$\Pr(|B(x)\cap Q_H|\geq q\beta d/2)\geq \nicefrac{1}{2}$.
The distance from $\pi(x)$ to $\pi(x')$ is at most $\nicefrac{\delta}{2}$, hence,
$B(x) \subset \{v\in V_G:\|\varphi(v)-\varphi(\pi(x'))\|^2 \leq \delta\}$. Thus, if $|B(x)\cap Q_G(x)|\geq q\beta d/2$,
then  $\pi(x)\in \Xi'(x)$. We get 
\begin{equation}\label{eq:boundPi1}
\Pr(\pi(x) \in \Xi'(x)\Given (x,x')\in \calS_i)\geq \frac{1}{2}.
\end{equation}

We now estimate $\Pr(|\Xi''(x)|\leq 8\eta n\Given (x,x')\in \calS_i)$. By Markov's inequality
the probability that $u\in \Xi''(x)$ (over a random $Q_G$ and fixed $x'$) is bounded by 
$$\Pr(u \in \Xi''(x)) \leq \min\Big\{\frac{|N_G(u)\cap \Ball_{\varphi} (\pi(x'), 2\delta)|}{\beta d/2},1\Big\}.$$
The expected size of $\Xi''(x)$ is bounded by 
$$\Exp\, |\Xi''(x)| \leq \frac{1}{\beta d}\;\sum_{u\in V_G} \min\big\{2|N_G(u)\cap \Ball_{\varphi} (\pi(x'), 2\delta)|,\beta d\big\}.$$
Now if $(x,x')\in \calS_i$, then by the definition of $\calS$, $M_{\beta d}(\Ball_{\varphi} (\pi(x'), 2\delta))\leq \eta n$. Thus,
\begin{eqnarray*}
\sum_{u\in V_G} \min\{2|N_G(u) \cap \Ball_{\varphi} (\pi(x'), 2\delta)|,\beta d\}&\leq&
2\sum_{u\in V_G} \min\{|N_F(u) \cap \Ball_{\varphi} (\pi(x'), 2\delta)|,\beta d\} \\ 
&\equiv& 2M_{\beta d}(\Ball_{\varphi} (\pi(x'), 2\delta))\\
&\leq& 2\eta \beta d n .
\end{eqnarray*}
We obtain the bound $\Exp\, |\Xi''(x)|\leq \eta n$. Applying Markov's inequality, we get
$$\Pr(|\Xi''(x)|\geq 4\eta n\Given (x,x')\in \calS)\leq \frac{1}{4}.$$
Combining this inequality with (\ref{eq:boundPi1}), we obtain the following bound:
$$\Pr(\pi(x)\in \Xi'(x) \text{ and }|\Xi''(x)|\leq 8\eta n\Given (x,x')\in \calS)\geq \frac{1}{4},$$
which is equivalent to 
$$\Pr(\pi(x)\in \Xi'(x);\; |\Xi''(x)|\leq 8\eta n;\;(x,x')\in \calS)\geq \frac{\Pr((x,x')\in \calS)}{4}.$$
We conclude that the expected size of $\calX$ is lower bounded by 
\begin{eqnarray*}
\Exp\;|\calX|&\geq& \sum_{x\in \calU_i}\frac{\Pr((x,x')\in \calS)}{4}\geq
\sum_{x\in \calU_i}  \frac{3|N_H(x)\cap \calS_i|}{16\cdot 2^{i+1}}
= \frac{|\calS_i|}{11\cdot 2^{i}}\\
&\geq&\frac{|\calS_i|}{11|\calE_i|}\;|\calU_i|=\frac{\gamma_i}{11}\;|\calU_i|.
\end{eqnarray*}
Since $|\calX|$ never exceeds $|\calU_i|$, we get $\Pr(|\calX|\geq (\gamma_i/22)\,|\calU|)\geq \gamma_i/22\geq 1/(44 D_{ARV})$.
This finishes the proof of the lemma.
\end{proof}

We now continue the proof of the Main Structural Theorem. We fix sets $Q_G$ and $Q_H$ satisfying the conditions of Lemma~\ref{lem:existQ}. 
We embed all vectors $\varphi(u)$ and $\varphi(\pi(x))$ for $u\in Q_G$ and $x\in Q_H$ in
a low dimensional space using the Johnson---Lindenstrauss transform. We pick the dimension and scaling
in a such way that if $\|\varphi(u) - \varphi(\pi(x))\|^2\leq \delta$ then $d(u,x)\leq \delta$; if $\|\varphi(u) - \varphi(\pi(x))\|^2> 2\delta$ then $d(u,x)> \delta$,
where $d(u,x)$ are the distances between the embedded vectors. In Lemma~\ref{lem:KP-d} in Appendix, we show
that such a distance function $d$ can be encoded using $O(qn \log n)$ bits (the function $d$ may not satisfy triangle inequalities). In other words, $KP(d(\cdot,\cdot))=O(q n\log n)$. 
We define the set $\Xi(x)$ as follows:
$$\Xi(x)=\{u\in V_G: |N_G(u)\cap \{v\in Q_G: d(v,x') \leq \delta\}| \geq q\beta d/2\}.$$
By our choice of $d$, we have $\Xi'(x)\subset \Xi(x)\subset \Xi''(x)$. Particularly, for $x\in \calX$, $\pi(x)\in \Xi(x)$ and $|\Xi(x)|\leq 4\eta n$.
Note that $\Xi(x)$ depends only on the graphs $G$, $H$, the sets $Q_G$, $Q_H$ and the distance function $d$. It does not depend on
the permutation $\pi$.

We now show how to encode the pair $(\calX,\pi|_\calX)$ (here $\pi|_{\calX}$ is the restriction of $\pi$ to $\calX$). We first encode $i$ using $\roundup{\bLog\bLog n}$ bits, 
then we encode $|\calX|$ using $\roundup{\bLog n}$ bit. We encode $\calX\subset \calU_i$ using 
$\roundup{\bLog\binom{|\calU_i|}{|\calX|}|}\leq \bLog\binom{|\calU_i|}{|\calX|}|+1$
bits. We can do so, since there are $\binom{|\calU_i|}{|\calX|}|$ subsets of $\calU_i$ of size
$|\calX|$. Note, that
$$\bLog\binom{|\calU_i|}{|\calX|} + 1\leq \bLog \left( \frac{e|\calU_i|}{|\calX|} \right)^{|\calX|} + 1
<|\calX|\bLog \left(\frac{|\calU_i|}{|\calX|}\right) + 2|\calX|.$$
We denote $\gamma' = |\calX|/|\calU_i|$. By Lemma~\ref{lem:existQ}, $\gamma'\geq c\gamma_i \geq c \gamma/2$.
We encode $Q_G$ and $Q_H$ using $O(qn\log n)$ bits. Finally, for every $x\in \calX$, we encode the index of $\pi(x)$ in $\Xi(x)$ using
$\roundup{\bLog |\Xi(x)|}\leq \bLog{(16\eta n)}$ bits. Altogether we use at most
\begin{align*}
KP((\calX,\pi|_{\calX}))&\leq \roundup{\bLog\bLog n} + \roundup{\bLog n} + (|\calX|\bLog (1/\gamma') + 2|\calX|) + O(qn\log n)
+ |\calX| \cdot \bLog (16\eta n)\\
&\leq |\calX|\bLog (16\eta n/\gamma') + 2|\calX| + 2\roundup{\bLog n} + O(qn\log n)
\end{align*}
bits. Lemma~\ref{lem:conditional-kp}, which we prove below, shows that we can extend the encoding of $(\calX,\pi|_\calX)$ to the encoding  
of $\pi$ using extra $\bLog|\Pi_{LR}| - |\calX|\bLog n +3|\calX|$ bits. Hence,
the total number of bits we need is
\begin{align*}
KP(\pi)&\leq KP(\calX\Given (\calX,\pi|_\calX)) + KP((\calX,\pi|_\calX)) + O(1)\\
&\leq \big(\bLog|\Pi_{LR}| - |\calX|\bLog n +3|\calX|\big) + \big(|\calX|\bLog (16\eta n/\gamma') + 2|\calX| + 2\roundup{\bLog n}+O(qn\log n)\big)\\
&\leq \bLog|\Pi_{LR}| + |\calX|\bLog (16\eta/\gamma') + 5|\calX| + 2\roundup{\bLog n}+O(qn\log n)\\
&\leq \bLog|\Pi_{LR}| - |\calX|\bLog (c' \gamma/\eta) + O(qn\log n) + 2\roundup{\bLog n}.
\end{align*}
for some constant $c'$. Here we used that $\gamma'\geq c\gamma_i\geq c\gamma/2$. To finish the proof of the Main Structural Theorem,
we need to show that $|\calX|\geq \Omega (\max\{qn\log n,\log n\})$ (see Claim~\ref{cl:sizeX}). This would imply that for a sufficiently large constant $K$, if $\gamma/\eta >K$,
then 
$$|\calX|\bLog (c' \gamma/\eta) - O(qn\log n) -2\roundup{\bLog n} > |\calX|\bLog (c'K)  - O(qn\log n) -2\roundup{\bLog n}> 2\bLog n.$$
Hence , if $\gamma/\eta >K$, then
$$KP(\pi)\leq \bLog|\Pi_{LR}| - 2\bLog n.$$
Therefore, $\gamma/\eta >K$ with probability at most $n^{-2}$ (since the number of $\pi$'s with 
prefix Kolmogorov complexity smaller than $\bLog|\Pi_{LR}| - 2\bLog n$ is 
at most $2^{\bLog|\Pi_{LR}| - 2\bLog n}\leq |\Pi_{LR}|/n^{2}$).

\begin{claim}\label{cl:sizeX} The following bound holds.
$$|\calX|\geq \Omega (\max\{qn\log n,\log n\}).$$
\end{claim}
\begin{proof}
We lower bound the size of $\calX$ as follows. By Lemma~\ref{lem:existQ}, $|\calX| = \Omega(\gamma' |\calU_i|) \geq \Omega(\gamma_i |\calU_i|)$.
Then, 
$$\gamma_i |\calU_i| \geq \frac{\gamma_i |\calE_i|}{2^{i+1}}= \frac{\gamma_i \lambda_i dn}{2^{i+1}}\geq \frac{\gamma dn}{2^{i+2}\bLog n}.$$
Here we used that $i\in I$, and thus $\lambda_i\gamma_i\geq \gamma/(2\bLog n)$. Since $2^i\leq n$, $d=\Omega (\log^3n)$ and 
$\gamma \geq \eta = \Omega(1/D_{ARV}) = \Omega(1/\sqrt{\log n})$, 
we have $\gamma_i|\calU_i|\geq\Omega(\log^{\nicefrac{3}{2}}n)$. Again using the bounds on $d$ and $\gamma$, we get
$\gamma d/\bLog n\geq D_n\bLog n$ and 
$$ \gamma_i |\calU_i| \geq  \frac{\gamma dn}{2^{i+3}\bLog n}\geq  \Big(\frac{D_n}{2^i}\Big) n\bLog n= \Omega(qn\log n).$$
\end{proof}
\end{proof}

\begin{lemma}\label{lem:conditional-kp}
Let $\pi\in \Pi_{LR}$ be a bijection from $V_H$ to $V_G$ mapping $L_H$ to $L_G$ and $R_H$ to $R_G$. Consider a subset $\calX\subset V_H$. Then,
$$KP(\pi)\leq KP((\pi|_{\calX},\calX)) + \bLog|\Pi_{LR}| - |\calX|\bLog n +3|\calX|+O(1).$$
That is, if the pair $(\pi|_{\calX},\calX)$ can be encoded using $KP((\pi|_{\calX},\calX))$ bits, then $\pi$ can be encoded using 
$\bLog|\Pi_{LR}| - |\calX|\bLog n +3|\calX|+O(1)$ bits.
\end{lemma}
\begin{proof}
We first encode $\calX$ and $\pi|_{\calX}$ using $KP((\pi|_{\calX},\calX))$ bits. Then, we encode the restriction $\pi|_{{V_H} \setminus \calX}$. To do so,
we split the set $V_H\setminus \calX$ into two subsets $\bar\calX_L = L_H\setminus \calX$ and $\bar\calX_R = R_H\setminus \calX$. 
Let \rule{0pt}{2ex}$m_L=|\bar\calX_L|$ and $m_R=|\bar\calX_R|$.
The restrictions of $\pi$ to $\bar\calX_L$ and to $\bar\calX_R$ are bijections from $\bar\calX_L$ to
$L_G\setminus \pi(V_H)$ and
from $\bar\calX_R$ to $R_G\setminus \pi(V_H)$ respectively. Hence, we can encode $\pi|_{\bar\calX_R}$
and $\pi|_{\bar\calX_R}$ using $\roundup{\bLog{m_L!}}$ and $\roundup{\bLog{m_R!}}$ bits (given $\calX$ and $\pi(\calX)$).
In other words,
\begin{align*}
KP(\pi)&\leq KP ((\pi|_{\calX},\calX)) + KP(\pi|_{\calX_L}\Given \calX, \pi(\calX)) + KP(\pi|_{\calX_R}\Given \calX, \pi(\calX)) + O(1)\\
&\leq KP ((\pi|_{\calX},\calX)) + \bLog (m_L!) + \bLog (m_R!) + O(1).
\end{align*}
Here $KP(\pi|_{\calX_L}\Given \calX, \pi(\calX))$ and $KP(\pi|_{\calX_R}\Given \calX, \pi(\calX))$ are conditional prefix Kolmogorov complexities
of $\pi|_{\calX_L}$ and $\pi|_{\calX_R}$ given $(\calX, \pi(\calX))$. We estimate $\bLog{m_L!} + \bLog{m_R!}$ using Stirling's approximation:
\begin{align*}
\bLog{m_L!} &= \bLog ((n/2)!) - \sum_{i=m_L+1}^{n/2} \bLog i\leq \bLog (n/2)! - \int_{m_L}^{n/2}\bLog x \,dx\\
&= \bLog ((n/2)!) - \frac{n}{2} \bLog\frac{n}{2} + m_L \bLog m_L + \bLog e \cdot (\frac{n}{2}-m_L)\\
&\leq \bLog  ((n/2)!) -(\frac{n}{2}-m_L)(\bLog\frac{n}{2} - \nicefrac{3}{2}).
\end{align*}
Thus,
$$\bLog(m_L!) + \bLog(m_R!)\leq 
2\bLog ((n/2)!) - (n-m_L-m_R)(\bLog\frac{n}{2} - \nicefrac{3}{2})
=\bLog |\Pi_{LR}| - |\calX| (\bLog n - \nicefrac{5}{2}).
$$
\end{proof}

\subsection{Proof of the Structural Properties 2 and 3} \label{sec:StructProp23}
Structural Property 3 immediately follows from the assumption we made in Section~\ref{sec:prelim}, and thus it is always satisfied. We now show how to derive Property 2 from the Main Structural Theorem.

\begin{theorem}
Structural Property 2 holds with probability $1-o(1)$.
\end{theorem}
\begin{proof}
Consider a subset of edges $E'_G$  of the graph $G$. An edge $(u,v)$ belongs to $E'_G$ if one of the endpoints, $u$ or $v$,
belongs to $\VGleqA$ (see (\ref{eq:VGleqA}) for the definition of $\VGleqA$) i.e.,
$$E'_G = \{(u,v)\in E_G: u\in \VGleqA\}.$$
The degrees of vertices in $\VGleqA$ are upper bounded by $\alpha d$. The size of $\VGleqA$ is bounded by $n/\alpha$. Thus, the 
set $E'_G$ has at most $dn$ edges. Let $G'=(V_G,E_G')$ be the graph with edges $E_G'$. We now apply the Main Structural Theorem 
to the graph $F' = H\boxplus_{\pi^{-1}}G'$ i.e., we switch around $G$ and $H$. This is possible, since the graph $G'$ has at most $dn$ edges.
The theorem implies that for every $\varphi$ and $\eta=2^{-i}$,
there are at most $K \eta dn$ edges satisfying the conditions below:
\begin{enumerate}
\item $(u,v)$ is a $\nicefrac{\delta}{2}$-short edge in $E'_G$
\item $M^{F'}_{\beta d}(\Ball_{\varphi}(v,  2\delta))\leq \eta n$, here $M^{F'}_{\beta d}(\cdot)$ is defined as 
$M_{\beta d}(\cdot)$ in~(\ref{def:M}) but only for the graph $F'$.
\item $\short_{\varphi, \nicefrac{\delta}{2}} (u, H)\geq \max\{\beta d, \deg(u,H)/D_n\}$.
\end{enumerate}
The first condition is equivalent to conditions 1 and 4 of Property 2. The second condition is less restrictive than the second condition of Property 2,
because $M^{F'}_{\beta d}(\Ball_{\varphi}(u,  2\delta))\leq M_{\beta d}(\Ball_{\varphi}(u,  2\delta))$ (since all edges of $\pi (F')$
are also edges of $F$). Finally, the third condition above is equivalent to the third condition of Property 2, because
$$\max\{\beta d, \deg(u,H)/D_n\}\leq \max\{\beta d, \alpha d /D_n\} = \beta D_n.$$
\end{proof}

\subsection{Proof of the Structural Property 4} \label{sec:StructProp4}
We first prove a simple technical lemma.
\begin{lemma}\label{lem:prop4holds}
Consider a vertex weighted graph $H=(V_H,E_H)$ with weights $c_y:V_H\to [0,1]$. Let $T\subset V_H$ be a subset of $V_H$ of size $n'=n/2$. 
Fix an integer $k\leq n'/12$. Let $S$ be a random subset of $T$ of size $k$. Then, 
$$\Pr\Big\{\forall x\in V_H, \sum_{\substack{(x,y)\in E_H\\y\in V_H\setminus S}} c_y\leq \frac{4k}{n'}\sum_{\substack{(x,y)\in E_H\\y\in S}} c_y+2\log n\Big\}\geq 1 - o(1).$$
\end{lemma}
\begin{proof} We first bound the probability that 
\begin{equation}\label{ineq:DegInZ}
\sum_{\substack{(x,y)\in E_H\\y\in S}} c_y \geq \frac{4k}{n'}\sum_{\substack{(x,y)\in E_H\\y\in V_H\setminus S}}c_y+2\log n.
\end{equation}
for a fixed vertex $x\in V_H$. Let 
$$\tilde c_y=
\begin{cases}
c_y,&\text{if } (x,y)\in E_H;\\
0,&\text{otherwise.}
\end{cases}$$
Write the inequality~(\ref{ineq:DegInZ}) in terms of $\tilde c_y$:
$$
\sum_{\substack{y\in S}} \tilde c_y \geq \frac{4k}{n'}\sum_{y\in V_H\setminus S} \tilde c_y +2\log n
\geq \frac{4k}{n'}\sum_{y\in T\setminus S} \tilde c_y +2\log n,
$$
or, equivalently, 
$$
(1+\frac{4k}{n'})\sum_{\substack{y\in S}} \tilde c_y \geq \frac{4k}{n'}\sum_{y\in T} \tilde c_y +2\log n.
$$
We denote $\mu = \frac{k}{n'}\,\sum_{y\in T} \tilde c_y$. Then, 
\begin{align*}
\Pr\big((\ref{ineq:DegInZ})\text{ holds}\big) &\leq \Pr\Big(\sum_{\substack{y\in S}} \tilde c_y \geq \frac{n'}{n'+4k}(4\mu +2\log n)\Big)
\leq \Pr\Big(\sum_{\substack{y\in S}} \tilde c_y \geq 3\mu +\nicefrac{3}{2}\,\log n\Big)\\
&\leq \Pr\Big(\sum_{\substack{y\in S}} \tilde c_y -\mu\geq 2\mu +\nicefrac{3}{2}\,\log n\Big)
\end{align*}
here we used that $k/n'\leq 1/12$ and $n'/(n'+4k)\geq 3/4$. Let $S'$ be a random multiset sampled from $T$ with replacement such 
that each $y$ belongs to $S'$ with probability $k/n'$. Hoeffding \cite[Theorem 4]{hoeffding1963}
showed that 
$$\Pr\Big(\sum_{\substack{y\in S}} \tilde c_y -\mu\geq 2\mu +\nicefrac{3}{2}\,\log n\Big)
\leq 
\Pr\Big(\sum_{\substack{y\in S'}} \tilde c_y -\mu\geq 2\mu +\nicefrac{3}{2}\,\log n\Big).
$$
Observe that $\sum_{\substack{y\in S'}} \tilde c_y =\mu$; $\Var[\sum_{\substack{y\in S'}} \tilde c_y]\leq \mu$; and $c_y\in[0,1]$ for all $y\in T$. 
Thus, by Bernstein's  inequality,
\begin{eqnarray*}
\Pr\Big(\sum_{\substack{y\in S'}} \tilde c_y -\mu\geq 2\mu +\nicefrac{3}{2}\,\log n\Big)
&\leq& \exp\Big(-\frac{(2\mu +  \frac{3}{2} \log n)^2}{2\mu + \frac{2}{3}(2\mu +\frac{3}{2}\,\log n)}\Big)\\
&\leq& \exp\Big(-\frac{(2\mu +  \frac{3}{2} \log n)^2}{4\mu + \log n}\Big)\\
&\leq& \exp(-2\log n) = n^{-2}.
\end{eqnarray*}
By the union bound, the probability that the inequality~(\ref{ineq:DegInZ}) holds for some $x\in V_H$ is at most $1/n$. This 
concludes the proof.
\end{proof}
\

\begin{theorem}\label{thm:prop4holds}
Structural Property 4 holds with probability $1-o(1)$. I.e., for every $u\in V_F$,
\begin{equation}\label{eq:prop4ineq}
\sum_{\substack{v:(x,y)\in \pi E_H\\v\in \VHgeqB \cap \VGleqA}} \frac{\beta d}{\deg(v, H)}\leq 
\frac{8}{\alpha}\sum_{\substack{v: (x,y)\in \pi E_H\\v\in \VHgeqB\setminus \VGleqA}} \frac{\beta d}{\deg(v, H)}+4\log n.
\end{equation}
\end{theorem}
\begin{proof}
Consider two sets $S=\pi^{-1}(\VGleqA)$ and $U = \pi^{-1}(\VHgeqB) = \{y\in V_H:\deg (y, H)\geq \beta d\}$. The vertices of $S$ and $U$ belong to the graph $H$ and not to the graph $F$. We slightly abuse notation to denote by $\deg (y, H)$ the degree of $y$ in $H$ (previously we used this notation for $u\in V_F$). 
Note that the set $S$ is a random though non-completely uniform subset of $V_H$, but $U$ is not a random subset and does not depend on $\pi$.
We rewrite~(\ref{eq:prop4ineq}) as follows: for every $x\in V_H$,
\begin{equation}\label{eq:prop4ineqinH}
\sum_{\substack{y:(x,y)\in E_H\\y\in U \cap S}} \frac{\beta d}{\deg(y, H)}\leq 
\frac{8}{\alpha}\sum_{\substack{y: (x,y)\in E_H\\y\in U\setminus S}} \frac{\beta d}{\deg(v, H)}+4\log n.
\end{equation}
We split the set $S$ into two sets $S_L=S\cap L_H$ and $S_R=S\cap R_H$. 
We show that for every $x\in V_H$, the following two inequalities hold:
\begin{align}
\sum_{\substack{v:(x,y)\in \pi E_H\\v\in U_L \cap S_L}} \frac{\beta d}{\deg(v, H)}&\leq 
\frac{8}{\alpha}\sum_{\substack{v: (x,y)\in E_H\\v\in U_L\setminus S_L}} \frac{\beta d}{\deg(v, H)}+2\log n;
\label{eq:prop4Left}\\
\sum_{\substack{v:(x,y)\in \pi E_H\\v\in U_R \cap S_R}} \frac{\beta d}{\deg(v, H)}&\leq 
\frac{8}{\alpha}\sum_{\substack{v: (x,y)\in  E_H\\v\in U_R\setminus S_R}} \frac{\beta d}{\deg(v, H)}+2\log n,
\end{align}
which together imply~(\ref{eq:prop4ineqinH}) and~(\ref{eq:prop4ineq}). These inequalities are the same up to 
renaming of $L$ and $R$. So we consider only the first inequality. We set the weight of each vertex $y\in L_H$ to be
$$c_y=
\begin{cases}
\frac{\beta d}{\deg(y,H)},&\text{if } y\in U;\\
0,&\text{otherwise}.
\end{cases}
$$
Note that for all $y\in V_H$, we have $c_y\in[0,1]$ and 
$|S_L|=|\VGleqA\cap L_G|\leq |\VGleqA|\leq n/\alpha \leq n/24$. The set $S_L$ is a random subset of $L_H$ of size $|\VGleqA\cap L|$. Hence, by 
Lemma~\ref{lem:prop4holds}, 
$$\Pr\Big\{\forall x\in V_H, \sum_{\substack{(x,y)\in E_H\\y\in V_H\setminus S}} c_y\leq \frac{8|\VGleqA\cap L|}{n}\sum_{\substack{(x,y)\in E_H\\y\in S}} c_y+2\log n\Big\}\geq 1 - o(1).$$
This inequality implies~(\ref{eq:prop4Left}) since $8|S_L|/n\leq 8/\alpha$.
\end{proof}

\bibliographystyle{plain}
\bibliography{piemodel}

\appendix
\section{Proof of Lemma~\ref{lem:simpleBalancedCut}}
\begin{lemma}\label{lem:simpleBalancedCut}
Suppose that $|\VGleqA| \geq n/\alpha$. Consider the following algorithm for Balanced Cut: sort all vertices according to their degree in $F$,
let $L'$ be the $\roundup{n/(3\alpha)}$ vertices with least degrees and $R' = V_F\setminus L'$, return the cut $(L',R')$.
The algorithm return a $\Theta(1)$-balanced cut of cost $O(dn)$ with high probability.
\end{lemma}
\begin{proof}
The cut $(L',R')$ is $1/(3\alpha)$ balanced, as required. We show that its cost is $\Theta(dn)$ w.h.p. Note that at least half of all vertices in $H$ have degree at most $2d$ by Markov's inequality. The permutation $\pi$ maps at least a $|\VGleqA| /n $ fraction of them to $\VGleqA$ in expectation. Thus the fraction of vertices in $F$ with degree at most $(\alpha +2) d$ is at least $(1/2) \cdot (|\VGleqA|/n) \geq 1/(2\alpha)$ in expectation. With high probability, there are at least $n/(3\alpha)$ vertices in $F$ of degree at most $(\alpha +2) d$. 
Then all vertices in $L$ have degrees at most $(\alpha +2) d$. Thus the cost of the cut $(L',R')$ is at most 
$(\alpha +2) d \times |L'| = (\alpha +2) d \cdot \rounddown{n/(3\alpha)} = O(dn)$.
\end{proof}

\section{Min Cut in Damage Control}\label{appendix:damage-control:min-cut}
In the Damage Control procedure, we solve a minimum cut problem in order to find $Y$ that maximizes~(\ref{eq:flowproblem}). Let us 
verify that the solution we obtain indeed maximizes $Y$. Consider an arbitrary cut 
$$(\{\text{``source''}\}\cup Y, \{\text{``sink''}\}\cup \bar Y).$$
This cut cuts all edges going from $Y$ to $\bar Y$. The capacity 
of these edges is $2|E_{F_3(t)}(Y,\bar Y)|$. Then, it cuts all edges going from the source to $\bar{Y}$. The capacity of these
edges equals $\budget(\bar{Y})$. Finally, it cuts all edges going from $Y$ to the sink. The capacity of these edges 
equals $2\beta d |Y|$. Thus, the total size of the cut equals
$$2|E_{F_3(t)}(Y,\bar Y)| + \budget(\bar{Y}) + 2\beta d |Y| = 2|E_{F_3(t)}(Y,\bar Y)| + \budget (V_{F_3(t)}) - \budget(Y) + 2\beta d |Y|.$$
The term $\budget (V_{F_3(t)})$ does not depend on the cut. Hence, the cut is minimized, when the expression~(\ref{eq:flowproblem})
is maximized.

\section{Proof of Lemma~\ref{lem:KP-d}}
We show that there exists a distance function $d:Q_G\times Q_H\to\bbR^+$ of small complexity that approximately preserves
balls of radius $\delta$.
\begin{lemma}\label{lem:KP-d}
There exists a function $d:Q_G\times Q_H\to\bbR^+$ such that 
$$KP(d\Given Q_G,Q_H\})=O(\max\{|Q_G|,|Q_H|\}\log n)$$ and for every $x\in Q_H$,
$$\{u\in Q_G: \|\varphi(u) - \varphi(\pi(x))\|^2  \leq \delta \}\subset \{u\in Q_G: d(u,x)\leq \delta\}\subset \{u\in Q_G: \|\varphi(u) - \varphi(\pi(x))\|^2  \leq 2\delta \}.$$
\end{lemma}
\begin{proof}
The proof of the lemma is very standard. We embed all vectors in $\varphi(Q_G)$ and $\varphi(Q_H)$ in a lower dimensional space 
via the Johnson---Lindenstrauss transform and then replace the embedded vectors with vectors in sufficiently dense low dimensional 
epsilon net. Instead of presenting the details we use a lemma from our previous work \cite{MMV-ITCS}.

\begin{lemma}[Lemma~2.7 in~\cite{MMV-ITCS}]\label{lem:classW}
For every $m$ and $\varepsilon \in (0,1)$, there exists a set of matrices
$\calA$ of size at most $|\calA|\leq \exp(O(\frac{m \log m}{2\varepsilon^2}))$ such that:
for every collection of vectors $L(1),\dots, L(m)$, $R(1),\dots R(m)$ with
$\|L(i)\| \leq 1$, $\|R(j)\| \leq 1$ and $\langle L(i), R(j)\geq 0$, 
there exists $A\in \calA$ satisfying for every $u$ and $x$:
$$ a(u,x)\leq \langle L(u), R(x)\rangle \leq a(u,x) + \gamma;$$
$$ a(u,x)\in [0,1].$$
\end{lemma}

Let $m=\max\{|Q_G|,|Q_H|\}$ and $\varepsilon = \delta/2$. We pick $\calA$ as in the lemma above. The set $\calA$ depends
only on $m$ and $\varepsilon$. We find a matrix $a\in \calA$ such that 
$a(u,x)\leq \langle \varphi(u) , \varphi(\pi(x))\rangle \leq a(u,x) + \varepsilon$ and let $d(u,x) = (1 - 2a(u,x))/2$. The complexity 
$KP(d\Given Q_G,Q_H)$ is at most $\roundup{\bLog |\calA|}=O(m \log m)$ since $d$ can be reconstructed from 
$a$, and $a$ is chosen among $\exp(O(\frac{m \log m}{2\varepsilon^2}))$ possible matrices. 
If $\|\varphi(u) - \varphi(\pi(x))\|^2\leq \delta$, then $\langle \varphi(u) , \varphi(\pi(x))\rangle\geq (1-\delta)/2$. Hence,
$a(u,x)\geq (1-2\delta)/2$ and $d(u,x)\leq \delta$. If $\|\varphi(u) - \varphi(\pi(x))\|^2>2\delta$, then
$a(u,x)\leq \langle \varphi(u) , \varphi(\pi(x))\rangle< (1-2\delta)/2$, and $d(u,x)>\delta$.
\end{proof}

\end{document}